\documentclass[a4paper,12pt,DIV13]{scrartcl}%

\usepackage{amsfonts}
\usepackage{amsmath}
\usepackage{amssymb}
\usepackage{graphicx}%
\newtheorem{theorem}{Theorem}

\newtheorem{definition}[theorem]{Definition}

\newtheorem{lemma}[theorem]{Lemma}

\newtheorem{proposition}[theorem]{Proposition}

\newenvironment{proof}[1][Proof]{\noindent\textbf{#1.} }{\ \rule{0.5em}{0.5em}}

\makeatletter
 
\def\diagram{\m@th\leftwidth=\z@ \rightwidth=\z@ \topheight=\z@
\botheight=\z@ \setbox\@picbox\hbox\bgroup}
 
\def\enddiagram{\egroup\wd\@picbox\rightwidth\unitlength
\ht\@picbox\topheight\unitlength \dp\@picbox\botheight\unitlength
\hskip\leftwidth\unitlength\box\@picbox}
 
\def\bfig{\begin{diagram}}
\def\efig{\end{diagram}}
\newcount\wideness \newcount\leftwidth \newcount\rightwidth
\newcount\highness \newcount\topheight \newcount\botheight
 
\def\ratchet#1#2{\ifnum#1<#2 \global #1=#2 \fi}
 
\def\putbox(#1,#2)#3{%
\horsize{\wideness}{#3} \divide\wideness by 2
{\advance\wideness by #1 \ratchet{\rightwidth}{\wideness}}
{\advance\wideness by -#1 \ratchet{\leftwidth}{\wideness}}
\vertsize{\highness}{#3} \divide\highness by 2
{\advance\highness by #2 \ratchet{\topheight}{\highness}}
{\advance\highness by -#2 \ratchet{\botheight}{\highness}}
\put(#1,#2){\makebox(0,0){$#3$}}}
 
\def\putlbox(#1,#2)#3{%
\horsize{\wideness}{#3}
{\advance\wideness by #1 \ratchet{\rightwidth}{\wideness}}
{\ratchet{\leftwidth}{-#1}}
\vertsize{\highness}{#3} \divide\highness by 2
{\advance\highness by #2 \ratchet{\topheight}{\highness}}
{\advance\highness by -#2 \ratchet{\botheight}{\highness}}
\put(#1,#2){\makebox(0,0)[l]{$#3$}}}
 
\def\putrbox(#1,#2)#3{%
\horsize{\wideness}{#3}
{\ratchet{\rightwidth}{#1}}
{\advance\wideness by -#1 \ratchet{\leftwidth}{\wideness}}
\vertsize{\highness}{#3} \divide\highness by 2
{\advance\highness by #2 \ratchet{\topheight}{\highness}}
{\advance\highness by -#2 \ratchet{\botheight}{\highness}}
\put(#1,#2){\makebox(0,0)[r]{$#3$}}}

\def\adjust[#1]{} % For compatibility
 
\newcount \coefa
\newcount \coefb
\newcount \coefc
\newcount\tempcounta
\newcount\tempcountb
\newcount\tempcountc
\newcount\tempcountd
\newcount\xext
\newcount\yext
\newcount\xoff
\newcount\yoff
\newcount\gap%
\newcount\arrowtypea
\newcount\arrowtypeb
\newcount\arrowtypec
\newcount\arrowtyped
\newcount\arrowtypee
\newcount\height
\newcount\width
\newcount\xpos
\newcount\ypos
\newcount\run
\newcount\rise
\newcount\arrowlength
\newcount\halflength
\newcount\arrowtype
\newdimen\tempdimen
\newdimen\xlen
\newdimen\ylen
\newsavebox{\tempboxa}%
\newsavebox{\tempboxb}%
\newsavebox{\tempboxc}%
 
\newdimen\w@dth
 
\def\setw@dth#1#2{\setbox\z@\hbox{\m@th$#1$}\w@dth=\wd\z@
\setbox\@ne\hbox{\m@th$#2$}\ifnum\w@dth<\wd\@ne \w@dth=\wd\@ne \fi
\advance\w@dth by 1.2em}
 
\def\t@^#1_#2{\allowbreak\def\n@one{#1}\def\n@two{#2}\mathrel
{\setw@dth{#1}{#2}
\mathop{\hbox to \w@dth{\rightarrowfill}}\limits
\ifx\n@one\empty\else ^{\box\z@}\fi
\ifx\n@two\empty\else _{\box\@ne}\fi}}
\def\t@@^#1{\@ifnextchar_{\t@^{#1}}{\t@^{#1}_{}}}
\def\to{\@ifnextchar^{\t@@}{\t@@^{}}}
 
\def\t@left^#1_#2{\def\n@one{#1}\def\n@two{#2}\mathrel{\setw@dth{#1}{#2}
\mathop{\hbox to \w@dth{\leftarrowfill}}\limits
\ifx\n@one\empty\else ^{\box\z@}\fi
\ifx\n@two\empty\else _{\box\@ne}\fi}}
\def\t@@left^#1{\@ifnextchar_{\t@left^{#1}}{\t@left^{#1}_{}}}
\def\toleft{\@ifnextchar^{\t@@left}{\t@@left^{}}}
 
\def\two@^#1_#2{\allowbreak
\def\n@one{#1}\def\n@two{#2}\mathrel{\setw@dth{#1}{#2}
\mathop{\vcenter{\lineskip\z@\baselineskip\z@
                 \hbox to \w@dth{\rightarrowfill}%
                 \hbox to \w@dth{\rightarrowfill}}%
       }\limits
\ifx\n@one\empty\else ^{\box\z@}\fi
\ifx\n@two\empty\else _{\box\@ne}\fi}}
\def\tw@@^#1{\@ifnextchar _{\two@^{#1}}{\two@^{#1}_{}}}
\def\two{\@ifnextchar ^{\tw@@}{\tw@@^{}}}
 
\def\tofr@^#1_#2{\def\n@one{#1}\def\n@two{#2}\mathrel{\setw@dth{#1}{#2}
\mathop{\vcenter{\hbox to \w@dth{\rightarrowfill}\kern-1.7ex
                 \hbox to \w@dth{\leftarrowfill}}%
       }\limits
\ifx\n@one\empty\else ^{\box\z@}\fi
\ifx\n@two\empty\else _{\box\@ne}\fi}}
\def\t@fr@^#1{\@ifnextchar_ {\tofr@^{#1}}{\tofr@^{#1}_{}}}
\def\tofro{\@ifnextchar^ {\t@fr@}{\t@fr@^{}}}

\def\mon{\mathop{\m@th\hbox to
      14.6\P@{\lasyb\char'51\hskip-2.1\P@$\arrext$\hss
$\mathord\rightarrow$}}\limits} % width of \epi
\def\leftmono{\mathrel{\m@th\hbox to
14.6\P@{$\mathord\leftarrow$\hss$\arrext$\hskip-2.1\P@\lasyb\char'50%
}}\limits} % width of \epi
\mathchardef\arrext="0200       % amr minus for arrow extension (see \into)

\def\settypes(#1,#2,#3){\arrowtypea#1 \arrowtypeb#2 \arrowtypec#3}
\def\settoheight#1#2{\setbox\@tempboxa\hbox{#2}#1\ht\@tempboxa\relax}%
\def\settodepth#1#2{\setbox\@tempboxa\hbox{#2}#1\dp\@tempboxa\relax}%
\def\settokens`#1`#2`#3`#4`{%
     \def\tokena{#1}\def\tokenb{#2}\def\tokenc{#3}\def\tokend{#4}}
\def\setsqparms[#1`#2`#3`#4;#5`#6]{%
\arrowtypea #1
\arrowtypeb #2
\arrowtypec #3
\arrowtyped #4
\width #5
\height #6
}
\def\setpos(#1,#2){\xpos=#1 \ypos#2}

\def\settriparms[#1`#2`#3;#4]{\settripairparms[#1`#2`#3`1`1;#4]}%
 
\def\settripairparms[#1`#2`#3`#4`#5;#6]{%
\arrowtypea #1
\arrowtypeb #2
\arrowtypec #3
\arrowtyped #4
\arrowtypee #5
\width #6
\height #6
}
 
\def\resetparms{\settripairparms[1`1`1`1`1;500]\width 500}%default values%
 
\resetparms
 
\def\mvector(#1,#2)#3{%%
\put(0,0){\vector(#1,#2){#3}}%
\put(0,0){\vector(#1,#2){26}}%
}
\def\evector(#1,#2)#3{{%%
\arrowlength #3
\put(0,0){\vector(#1,#2){\arrowlength}}%
\advance \arrowlength by-30
\put(0,0){\vector(#1,#2){\arrowlength}}%
}}
 
\def\horsize#1#2{%
\settowidth{\tempdimen}{$#2$}%
#1=\tempdimen
\divide #1 by\unitlength
}
 
\def\vertsize#1#2{%
\settoheight{\tempdimen}{$#2$}%
#1=\tempdimen
\settodepth{\tempdimen}{$#2$}%
\advance #1 by\tempdimen
\divide #1 by\unitlength
}
 
\def\putvector(#1,#2)(#3,#4)#5#6{{%
\ifnum3<\arrowtype
\putdashvector(#1,#2)(#3,#4)#5\arrowtype
\else
\ifnum\arrowtype<-3
\putdashvector(#1,#2)(#3,#4)#5\arrowtype
\else
\xpos=#1
\ypos=#2
\run=#3
\rise=#4
\arrowlength=#5
\ifnum \arrowtype<0
    \ifnum \run=0
        \advance \ypos by-\arrowlength
    \else
        \tempcounta \arrowlength
        \multiply \tempcounta by\rise
        \divide \tempcounta by\run
        \ifnum\run>0
            \advance \xpos by\arrowlength
            \advance \ypos by\tempcounta
        \else
            \advance \xpos by-\arrowlength
            \advance \ypos by-\tempcounta
        \fi
    \fi
    \multiply \arrowtype by-1
    \multiply \rise by-1
    \multiply \run by-1
\fi
\ifcase \arrowtype
\or \put(\xpos,\ypos){\vector(\run,\rise){\arrowlength}}%
\or \put(\xpos,\ypos){\mvector(\run,\rise)\arrowlength}%
\or \put(\xpos,\ypos){\evector(\run,\rise){\arrowlength}}%
\fi\fi\fi
}}
 
\def\putsplitvector(#1,#2)#3#4{%%
\xpos #1
\ypos #2
\arrowtype #4
\halflength #3
\arrowlength #3
\gap 140
\advance \halflength by-\gap
\divide \halflength by2
\ifnum\arrowtype>0
   \ifcase \arrowtype
   \or \put(\xpos,\ypos){\line(0,-1){\halflength}}%
       \advance\ypos by-\halflength
       \advance\ypos by-\gap
       \put(\xpos,\ypos){\vector(0,-1){\halflength}}%
   \or \put(\xpos,\ypos){\line(0,-1)\halflength}%
       \put(\xpos,\ypos){\vector(0,-1)3}%
       \advance\ypos by-\halflength
       \advance\ypos by-\gap
       \put(\xpos,\ypos){\vector(0,-1){\halflength}}%
   \or \put(\xpos,\ypos){\line(0,-1)\halflength}%
       \advance\ypos by-\halflength
       \advance\ypos by-\gap
       \put(\xpos,\ypos){\evector(0,-1){\halflength}}%
   \fi
\else \arrowtype=-\arrowtype
   \ifcase\arrowtype
   \or \advance \ypos by-\arrowlength
       \put(\xpos,\ypos){\line(0,1){\halflength}}%
       \advance\ypos by\halflength
       \advance\ypos by\gap
       \put(\xpos,\ypos){\vector(0,1){\halflength}}%
   \or \advance \ypos by-\arrowlength
       \put(\xpos,\ypos){\line(0,1)\halflength}%
       \put(\xpos,\ypos){\vector(0,1)3}%
       \advance\ypos by\halflength
       \advance\ypos by\gap
       \put(\xpos,\ypos){\vector(0,1){\halflength}}%
   \or \advance \ypos by-\arrowlength
       \put(\xpos,\ypos){\line(0,1)\halflength}%
       \advance\ypos by\halflength
       \advance\ypos by\gap
       \put(\xpos,\ypos){\evector(0,1){\halflength}}%
   \fi
\fi
}
 
\def\putmorphism(#1)(#2,#3)[#4`#5`#6]#7#8#9{{%
\run #2
\rise #3
\ifnum\rise=0
  \puthmorphism(#1)[#4`#5`#6]{#7}{#8}#9%
\else\ifnum\run=0
  \putvmorphism(#1)[#4`#5`#6]{#7}{#8}#9%
\else
\setpos(#1)%
\arrowlength #7
\arrowtype #8
\ifnum\run=0
\else\ifnum\rise=0
\else
\ifnum\run>0
    \coefa=1
\else
   \coefa=-1
\fi
\ifnum\arrowtype>0
   \coefb=0
   \coefc=-1
\else
   \coefb=\coefa
   \coefc=1
   \arrowtype=-\arrowtype
\fi
\width=2
\multiply \width by\run
\divide \width by\rise
\ifnum \width<0  \width=-\width\fi
\advance\width by60
\if l#9 \width=-\width\fi
\putbox(\xpos,\ypos){#4}%            %node 1
{\multiply \coefa by\arrowlength%      %node 2
\advance\xpos by\coefa
\multiply \coefa by\rise
\divide \coefa by\run
\advance \ypos by\coefa
\putbox(\xpos,\ypos){#5} }%
{\multiply \coefa by\arrowlength%      %label
\divide \coefa by2
\advance \xpos by\coefa
\advance \xpos by\width
\multiply \coefa by\rise
\divide \coefa by\run
\advance \ypos by\coefa
\if l#9%
   \putrbox(\xpos,\ypos){#6}%
\else\if r#9%
   \putlbox(\xpos,\ypos){#6}%
\fi\fi }%
{\multiply \rise by-\coefc%             %arrow
\multiply \run by-\coefc
\multiply \coefb by\arrowlength
\advance \xpos by\coefb
\multiply \coefb by\rise
\divide \coefb by\run
\advance \ypos by\coefb
\multiply \coefc by70
\advance \ypos by\coefc
\multiply \coefc by\run
\divide \coefc by\rise
\advance \xpos by\coefc
\multiply \coefa by140
\multiply \coefa by\run
\divide \coefa by\rise
\advance \arrowlength by\coefa
\ifcase\arrowtype
\or \put(\xpos,\ypos){\vector(\run,\rise){\arrowlength}}%
\or \put(\xpos,\ypos){\mvector(\run,\rise){\arrowlength}}%
\or \put(\xpos,\ypos){\evector(\run,\rise){\arrowlength}}%
\fi}\fi\fi\fi\fi}}

\newcount\numbdashes \newcount\lengthdash \newcount\increment
 
\def\howmanydashes{% Actually returns both number and length
\numbdashes=\arrowlength \lengthdash=40
\divide\numbdashes by \lengthdash
\lengthdash=\arrowlength
\divide\lengthdash by \numbdashes
\increment=\lengthdash
\multiply\lengthdash by 3
\divide\lengthdash by 5
}
 
\def\putdashvector(#1)(#2,#3)#4#5{%
\ifnum#3=0 \putdashhvector(#1){#4}#5
\else
\ifnum#2=0
\putdashvvector(#1){#4}#5\fi\fi}
 
\def\putdashhvector(#1,#2)#3#4{{%
\arrowlength=#3 \howmanydashes
\multiput(#1,#2)(\increment,0){\numbdashes}%
{\vrule height .4pt width \lengthdash\unitlength}
\arrowtype=#4 \xpos=#1
\ifnum\arrowtype<0 \advance\arrowtype by 7 \fi
\ifcase\arrowtype
\or \advance\xpos by 10
    \put(\xpos,#2){\vector(-1,0){\lengthdash}}
    \advance\xpos by 40
    \put(\xpos,#2){\vector(-1,0){\lengthdash}}
\or \advance \xpos by 10
    \put(\xpos,#2){\vector(-1,0){\lengthdash}}
    \advance\xpos by  \arrowlength
    \advance\xpos by  -50
    \put(\xpos,#2){\vector(-1,0){\lengthdash}}
\or \advance\xpos by 10
    \put(\xpos,#2){\vector(-1,0){\lengthdash}}
\or \advance\xpos by \arrowlength
    \advance\xpos by -\lengthdash
    \put(\xpos,#2){\vector(1,0){\lengthdash}}
\or {\advance\xpos by 10
    \put(\xpos,#2){\vector(1,0){\lengthdash}}}
    \advance\xpos by \arrowlength
    \advance\xpos by -\lengthdash
    \put(\xpos,#2){\vector(1,0){\lengthdash}}
\or \advance\xpos by \arrowlength
    \advance\xpos by -\lengthdash
    \put(\xpos,#2){\vector(1,0){\lengthdash}}
    \advance\xpos by -40
    \put(\xpos,#2){\vector(1,0){\lengthdash}}
   \fi
}}
 
\def\putdashvvector(#1,#2)#3#4{{%
\arrowlength=#3 \howmanydashes
\ypos=#2 \advance\ypos by -\arrowlength
\multiput(#1,#2)(0,\increment){\numbdashes}%
    {\vrule width .4pt height \lengthdash\unitlength}
\arrowtype=#4 \ypos=#2
\ifnum\arrowtype<0 \advance\arrowtype by 7 \fi
\ifcase\arrowtype
\or \advance\ypos by \arrowlength \advance\ypos by -40
    \put(#1,\ypos){\vector(0,1){\lengthdash}}
    \advance\ypos by -40
    \put(#1,\ypos){\vector(0,1){\lengthdash}}
\or \advance\ypos by 10
    \put(#1,\ypos){\vector(0,1){\lengthdash}}
    \advance\ypos by \arrowlength \advance\ypos by -40
    \put(#1,\ypos){\vector(0,1){\lengthdash}}
\or \advance\ypos by \arrowlength \advance\ypos by -40
    \put(#1,\ypos){\vector(0,1){\lengthdash}}
\or \advance\ypos by 10
    \put(#1,\ypos){\vector(0,-1){\lengthdash}}
\or \advance\ypos by 10
    \put(#1,\ypos){\vector(0,-1){\lengthdash}}
    \advance\ypos by \arrowlength \advance\ypos by -40
    \put(#1,\ypos){\vector(0,-1){\lengthdash}}
\or \advance\ypos by 10
    \put(#1,\ypos){\vector(0,-1){\lengthdash}}
    \advance\ypos by 40
    \put(#1,\ypos){\vector(0,-1){\lengthdash}}
\fi
}}
 
\def\puthmorphism(#1,#2)[#3`#4`#5]#6#7#8{{%
\xpos #1
\ypos #2
\width #6
\arrowlength #6
\arrowtype=#7
\putbox(\xpos,\ypos){#3\vphantom{#4}}%
{\advance \xpos by\arrowlength
\putbox(\xpos,\ypos){\vphantom{#3}#4}}%
\horsize{\tempcounta}{#3}%
\horsize{\tempcountb}{#4}%
\divide \tempcounta by2
\divide \tempcountb by2
\advance \tempcounta by30
\advance \tempcountb by30
\advance \xpos by\tempcounta
\advance \arrowlength by-\tempcounta
\advance \arrowlength by-\tempcountb
\putvector(\xpos,\ypos)(1,0)\arrowlength\arrowtype
\divide \arrowlength by2
\advance \xpos by\arrowlength
\vertsize{\tempcounta}{#5}%
\divide\tempcounta by2
\advance \tempcounta by20
\if a#8 %
   \advance \ypos by\tempcounta
   \putbox(\xpos,\ypos){#5}%
\else
   \advance \ypos by-\tempcounta
   \putbox(\xpos,\ypos){#5}%
\fi}}
 
\def\putvmorphism(#1,#2)[#3`#4`#5]#6#7#8{{%
\xpos #1
\ypos #2
\arrowlength #6
\arrowtype #7
\settowidth{\xlen}{$#5$}%
\putbox(\xpos,\ypos){#3}%
{\advance \ypos by-\arrowlength
\putbox(\xpos,\ypos){#4}}%
{\advance\arrowlength by-140
\advance \ypos by-70
\ifdim\xlen>0pt
   \if m#8%
      \putsplitvector(\xpos,\ypos)\arrowlength\arrowtype
   \else
   \putvector(\xpos,\ypos)(0,-1)\arrowlength\arrowtype
   \fi
\else
   \putvector(\xpos,\ypos)(0,-1)\arrowlength\arrowtype
\fi}%
\ifdim\xlen>0pt
   \divide \arrowlength by2
   \advance\ypos by-\arrowlength
   \if l#8%
      \advance \xpos by-40
      \putrbox(\xpos,\ypos){#5}%
   \else\if r#8%
      \advance \xpos by40
      \putlbox(\xpos,\ypos){#5}%
   \else
      \putbox(\xpos,\ypos){#5}%
   \fi\fi
\fi
}}
 
\def\putsquarep<#1>(#2)[#3;#4`#5`#6`#7]{{%
\setsqparms[#1]%
\setpos(#2)%
\settokens`#3`%
\puthmorphism(\xpos,\ypos)[\tokenc`\tokend`{#7}]{\width}{\arrowtyped}b%
\advance\ypos by \height
\puthmorphism(\xpos,\ypos)[\tokena`\tokenb`{#4}]{\width}{\arrowtypea}a%
\putvmorphism(\xpos,\ypos)[``{#5}]{\height}{\arrowtypeb}l%
\advance\xpos by \width
\putvmorphism(\xpos,\ypos)[``{#6}]{\height}{\arrowtypec}r%
}}
 
\def\putsquare{\@ifnextchar <{\putsquarep}{\putsquarep%
   <\arrowtypea`\arrowtypeb`\arrowtypec`\arrowtyped;\width`\height>}}
\def\square{\@ifnextchar< {\squarep}{\squarep
   <\arrowtypea`\arrowtypeb`\arrowtypec`\arrowtyped;\width`\height>}}
\def\squarep<#1>[#2`#3`#4`#5;#6`#7`#8`#9]{{%       %     #2------>#3
\setsqparms[#1]%                                   %      |       |
\diagram%                                          %      |       |
\putsquarep<\arrowtypea`\arrowtypeb`\arrowtypec`%  %    #7|       |#8
\arrowtyped;\width`\height>%                       %      |       |
(0,0)[#2`#3`#4`{#5};#6`#7`#8`{#9}]%                %      |       |
\enddiagram%                                       %      v       v
}}                                                 %     #4------>#5
\def\putptrianglep<#1>(#2,#3)[#4`#5`#6;#7`#8`#9]{{%
\settriparms[#1]%
\xpos=#2 \ypos=#3
\advance\ypos by \height
\puthmorphism(\xpos,\ypos)[#4`#5`{#7}]{\height}{\arrowtypea}a%
\putvmorphism(\xpos,\ypos)[`#6`{#8}]{\height}{\arrowtypeb}l%
\advance\xpos by\height
\putmorphism(\xpos,\ypos)(-1,-1)[``{#9}]{\height}{\arrowtypec}r%
}}
 
\def\putptriangle{\@ifnextchar <{\putptrianglep}{\putptrianglep
   <\arrowtypea`\arrowtypeb`\arrowtypec;\height>}}
\def\ptriangle{\@ifnextchar <{\ptrianglep}{\ptrianglep
   <\arrowtypea`\arrowtypeb`\arrowtypec;\height>}}
\def\ptrianglep<#1>[#2`#3`#4;#5`#6`#7]{{%%    %      #2----->#3
\settriparms[#1]%                             %      |      /
\diagram%                                     %      |     /
\putptrianglep<\arrowtypea`\arrowtypeb`%      %    #6|    /#7
\arrowtypec;\height>%                         %      |   /
(0,0)[#2`#3`#4;#5`#6`{#7}]%                   %      |  /
\enddiagram%%                                 %      v v
}}                                            %      #4
 
\def\putqtrianglep<#1>(#2,#3)[#4`#5`#6;#7`#8`#9]{{%
\settriparms[#1]%
\xpos=#2 \ypos=#3
\advance\ypos by\height
\puthmorphism(\xpos,\ypos)[#4`#5`{#7}]{\height}{\arrowtypea}a%
\putmorphism(\xpos,\ypos)(1,-1)[``{#8}]{\height}{\arrowtypeb}l%
\advance\xpos by\height
\putvmorphism(\xpos,\ypos)[`#6`{#9}]{\height}{\arrowtypec}r%
}}
 
\def\putqtriangle{\@ifnextchar <{\putqtrianglep}{\putqtrianglep
   <\arrowtypea`\arrowtypeb`\arrowtypec;\height>}}
\def\qtriangle{\@ifnextchar <{\qtrianglep}{\qtrianglep
   <\arrowtypea`\arrowtypeb`\arrowtypec;\height>}}
\def\qtrianglep<#1>[#2`#3`#4;#5`#6`#7]{{%%    %        #2----->#3
\settriparms[#1]%                             %         \      |
\width=\height                                %          \     |
\diagram%                                     %         #6\    |#7
\putqtrianglep<\arrowtypea`\arrowtypeb`%      %            \   |
\arrowtypec;\height>%                         %             \  |
(0,0)[#2`#3`#4;#5`#6`{#7}]%                   %              v v
\enddiagram%%                                 %               #4
}}
 
\def\putdtrianglep<#1>(#2,#3)[#4`#5`#6;#7`#8`#9]{{%
\settriparms[#1]%
\xpos=#2 \ypos=#3
\puthmorphism(\xpos,\ypos)[#5`#6`{#9}]{\height}{\arrowtypec}b%
\advance\xpos by \height \advance\ypos by\height
\putmorphism(\xpos,\ypos)(-1,-1)[``{#7}]{\height}{\arrowtypea}l%
\putvmorphism(\xpos,\ypos)[#4``{#8}]{\height}{\arrowtypeb}r%
}}
 
\def\putdtriangle{\@ifnextchar <{\putdtrianglep}{\putdtrianglep
   <\arrowtypea`\arrowtypeb`\arrowtypec;\height>}}
\def\dtriangle{\@ifnextchar <{\dtrianglep}{\dtrianglep
   <\arrowtypea`\arrowtypeb`\arrowtypec;\height>}}
\def\dtrianglep<#1>[#2`#3`#4;#5`#6`#7]{{%%    %                  / |
\settriparms[#1]%                             %                 /  |
\width=\height                                %              #5/   |#6
\diagram%                                     %               /    |
\putdtrianglep<\arrowtypea`\arrowtypeb`%      %              /     |
\arrowtypec;\height>%                         %             v      v
(0,0)[#2`#3`#4;#5`#6`{#7}]%                   %            #3----->#4
\enddiagram%%                                 %                #7
}}
 
\def\putbtrianglep<#1>(#2,#3)[#4`#5`#6;#7`#8`#9]{{%
\settriparms[#1]%
\xpos=#2 \ypos=#3
\puthmorphism(\xpos,\ypos)[#5`#6`{#9}]{\height}{\arrowtypec}b%
\advance\ypos by\height
\putmorphism(\xpos,\ypos)(1,-1)[``{#8}]{\height}{\arrowtypeb}r%
\putvmorphism(\xpos,\ypos)[#4``{#7}]{\height}{\arrowtypea}l%
}}
 
\def\putbtriangle{\@ifnextchar <{\putbtrianglep}{\putbtrianglep
   <\arrowtypea`\arrowtypeb`\arrowtypec;\height>}}
\def\btriangle{\@ifnextchar <{\btrianglep}{\btrianglep
   <\arrowtypea`\arrowtypeb`\arrowtypec;\height>}}
\def\btrianglep<#1>[#2`#3`#4;#5`#6`#7]{{%%   %              | \
\settriparms[#1]%                            %              |  \
\width=\height                               %            #5|   \#6
\diagram%                                    %              |    \
\putbtrianglep<\arrowtypea`\arrowtypeb`%     %              |     \
\arrowtypec;\height>%                        %              v      v
(0,0)[#2`#3`#4;#5`#6`{#7}]%                  %              #3----->#4
\enddiagram%%                                %                 #7
}}
 
\def\putAtrianglep<#1>(#2,#3)[#4`#5`#6;#7`#8`#9]{{%
\settriparms[#1]%
\xpos=#2 \ypos=#3
{\multiply \height by2
\puthmorphism(\xpos,\ypos)[#5`#6`{#9}]{\height}{\arrowtypec}b}%
\advance\xpos by\height \advance\ypos by\height
\putmorphism(\xpos,\ypos)(-1,-1)[#4``{#7}]{\height}{\arrowtypea}l%
\putmorphism(\xpos,\ypos)(1,-1)[``{#8}]{\height}{\arrowtypeb}r%
}}
 
\def\putAtriangle{\@ifnextchar <{\putAtrianglep}{\putAtrianglep
   <\arrowtypea`\arrowtypeb`\arrowtypec;\height>}}
\def\Atriangle{\@ifnextchar <{\Atrianglep}{\Atrianglep
   <\arrowtypea`\arrowtypeb`\arrowtypec;\height>}}
\def\Atrianglep<#1>[#2`#3`#4;#5`#6`#7]{{%%         %         /   \
\settriparms[#1]%                                  %        /     \
\width=\height                                     %     #5/       \#6
\diagram%                                          %      /         \
\putAtrianglep<\arrowtypea`\arrowtypeb`%           %     /           \
\arrowtypec;\height>%                              %    v             v
(0,0)[#2`#3`#4;#5`#6`{#7}]%                        %   #3------------>#4
\enddiagram%%                                      %          #7
}}
 
\def\putAtrianglepairp<#1>(#2)[#3;#4`#5`#6`#7`#8]{{%
\settripairparms[#1]%
\setpos(#2)%
\settokens`#3`%
\puthmorphism(\xpos,\ypos)[\tokenb`\tokenc`{#7}]{\height}{\arrowtyped}b%
\advance\xpos by\height
\puthmorphism(\xpos,\ypos)[\phantom{\tokenc}`\tokend`{#8}]%
{\height}{\arrowtypee}b%
\advance\ypos by\height
\putmorphism(\xpos,\ypos)(-1,-1)[\tokena``{#4}]{\height}{\arrowtypea}l%
\putvmorphism(\xpos,\ypos)[``{#5}]{\height}{\arrowtypeb}m%
\putmorphism(\xpos,\ypos)(1,-1)[``{#6}]{\height}{\arrowtypec}r%
}}
 
\def\putAtrianglepair{\@ifnextchar <{\putAtrianglepairp}{\putAtrianglepairp%
   <\arrowtypea`\arrowtypeb`\arrowtypec`\arrowtyped`\arrowtypee;\height>}}
\def\Atrianglepair{\@ifnextchar <{\Atrianglepairp}{\Atrianglepairp%
   <\arrowtypea`\arrowtypeb`\arrowtypec`\arrowtyped`\arrowtypee;\height>}}
 
\def\Atrianglepairp<#1>[#2;#3`#4`#5`#6`#7]{{%           %  #2a
\settripairparms[#1]%                         %           / | \
\settokens`#2`%                               %          /  |  \
\width=\height                                %       #3/  #4   \#5
\diagram%                                     %        /    |    \
\putAtrianglepairp                            %       /     |     \
<\arrowtypea`\arrowtypeb`\arrowtypec`%        %      v      v      v
\arrowtyped`\arrowtypee;\height>%             %     #2b---->#2c---->#2d
(0,0)[{#2};#3`#4`#5`#6`{#7}]%                 %         #6     #7
\enddiagram%%
}}
 
\def\putVtrianglep<#1>(#2,#3)[#4`#5`#6;#7`#8`#9]{{%
\settriparms[#1]%
\xpos=#2 \ypos=#3
\advance\ypos by\height
{\multiply\height by2
\puthmorphism(\xpos,\ypos)[#4`#5`{#7}]{\height}{\arrowtypea}a}%
\putmorphism(\xpos,\ypos)(1,-1)[`#6`{#8}]{\height}{\arrowtypeb}l%
\advance\xpos by\height
\advance\xpos by\height
\putmorphism(\xpos,\ypos)(-1,-1)[``{#9}]{\height}{\arrowtypec}r%
}}
 
\def\putVtriangle{\@ifnextchar <{\putVtrianglep}{\putVtrianglep
   <\arrowtypea`\arrowtypeb`\arrowtypec;\height>}}
\def\Vtriangle{\@ifnextchar <{\Vtrianglep}{\Vtrianglep
   <\arrowtypea`\arrowtypeb`\arrowtypec;\height>}}
\def\Vtrianglep<#1>[#2`#3`#4;#5`#6`#7]{{%%     %        #2------------->#3
\settriparms[#1]%                              %         \             /
\width=\height                                 %          \           /
\diagram%                                      %         #6\         /#7
\putVtrianglep<\arrowtypea`\arrowtypeb`%       %            \       /
\arrowtypec;\height>%                          %             \     /
(0,0)[#2`#3`#4;#5`#6`{#7}]%                    %              v   v
\enddiagram%%                                  %               #4
}}
 
\def\putVtrianglepairp<#1>(#2)[#3;#4`#5`#6`#7`#8]{{
\settripairparms[#1]%
\setpos(#2)%
\settokens`#3`%
\advance\ypos by\height
\putmorphism(\xpos,\ypos)(1,-1)[`\tokend`{#6}]{\height}{\arrowtypec}l%
\puthmorphism(\xpos,\ypos)[\tokena`\tokenb`{#4}]{\height}{\arrowtypea}a%
\advance\xpos by\height
\puthmorphism(\xpos,\ypos)[\phantom{\tokenb}`\tokenc`{#5}]%
{\height}{\arrowtypeb}a%
\putvmorphism(\xpos,\ypos)[``{#7}]{\height}{\arrowtyped}m%
\advance\xpos by\height
\putmorphism(\xpos,\ypos)(-1,-1)[``{#8}]{\height}{\arrowtypee}r%
}}
 
\def\putVtrianglepair{\@ifnextchar <{\putVtrianglepairp}{\putVtrianglepairp%
    <\arrowtypea`\arrowtypeb`\arrowtypec`\arrowtyped`\arrowtypee;\height>}}
\def\Vtrianglepair{\@ifnextchar <{\Vtrianglepairp}{\Vtrianglepairp%
    <\arrowtypea`\arrowtypeb`\arrowtypec`\arrowtyped`\arrowtypee;\height>}}
\def\Vtrianglepairp<#1>[#2;#3`#4`#5`#6`#7]{{%  %  #2a---->#2b---->#2c
\settripairparms[#1]%                          %   \      |      /
\settokens`#2`%                                %    \     |     /
\diagram%                                      %   #5\   #6    /#7
\putVtrianglepairp                             %      \   |   /
<\arrowtypea`\arrowtypeb`\arrowtypec`%         %       \  |  /
\arrowtyped`\arrowtypee;\height>%              %        v v v
(0,0)[{#2};#3`#4`#5`#6`{#7}]%                  %         #2d
\enddiagram%%
}}

\def\putCtrianglep<#1>(#2,#3)[#4`#5`#6;#7`#8`#9]{{%
\settriparms[#1]%
\xpos=#2 \ypos=#3
\advance\ypos by\height
\putmorphism(\xpos,\ypos)(1,-1)[``{#9}]{\height}{\arrowtypec}l%
\advance\xpos by\height
\advance\ypos by\height
\putmorphism(\xpos,\ypos)(-1,-1)[#4`#5`{#7}]{\height}{\arrowtypea}l%
{\multiply\height by 2
\putvmorphism(\xpos,\ypos)[`#6`{#8}]{\height}{\arrowtypeb}r}%
}}
 
\def\putCtriangle{\@ifnextchar <{\putCtrianglep}{\putCtrianglep
    <\arrowtypea`\arrowtypeb`\arrowtypec;\height>}}
\def\Ctriangle{\@ifnextchar <{\Ctrianglep}{\Ctrianglep
    <\arrowtypea`\arrowtypeb`\arrowtypec;\height>}}
\def\Ctrianglep<#1>[#2`#3`#4;#5`#6`#7]{{%%   %                / |
\settriparms[#1]%                            %             #5/  |
\width=\height                               %              /   |
\diagram%                                    %             v    |
\putCtrianglep<\arrowtypea`\arrowtypeb`%     %           #3     |#6
\arrowtypec;\height>%                        %             \    |
(0,0)[#2`#3`#4;#5`#6`{#7}]%                  %            #7\   |
\enddiagram%%                                %               \  |
}}                                           %                v v
\def\putDtrianglep<#1>(#2,#3)[#4`#5`#6;#7`#8`#9]{{%
\settriparms[#1]%
\xpos=#2 \ypos=#3
\advance\xpos by\height \advance\ypos by\height
\putmorphism(\xpos,\ypos)(-1,-1)[``{#9}]{\height}{\arrowtypec}r%
\advance\xpos by-\height \advance\ypos by\height
\putmorphism(\xpos,\ypos)(1,-1)[`#5`{#8}]{\height}{\arrowtypeb}r%
{\multiply\height by 2
\putvmorphism(\xpos,\ypos)[#4`#6`{#7}]{\height}{\arrowtypea}l}%
}}
 
\def\putDtriangle{\@ifnextchar <{\putDtrianglep}{\putDtrianglep
    <\arrowtypea`\arrowtypeb`\arrowtypec;\height>}}
\def\Dtriangle{\@ifnextchar <{\Dtrianglep}{\Dtrianglep
   <\arrowtypea`\arrowtypeb`\arrowtypec;\height>}}
\def\Dtrianglep<#1>[#2`#3`#4;#5`#6`#7]{{%%  %          | \
\settriparms[#1]%                           %          |  \#6
\width=\height                              %          |   \
\diagram%                                   %          |    v
\putDtrianglep<\arrowtypea`\arrowtypeb`%    %        #5|    #3
\arrowtypec;\height>%                       %          |    /
(0,0)[#2`#3`#4;#5`#6`{#7}]%                 %          |   /#7
\enddiagram%%                               %          |  /
}}                                          %          v v
\def\setrecparms[#1`#2]{\width=#1 \height=#2}%
\def\recursep<#1`#2>[#3;#4`#5`#6`#7`#8]{{\m@th
\width=#1 \height=#2
\settokens`#3`
\settowidth{\tempdimen}{$\tokena$}
\ifdim\tempdimen=0pt
  \savebox{\tempboxa}{\hbox{$\tokenb$}}%
  \savebox{\tempboxb}{\hbox{$\tokend$}}%
  \savebox{\tempboxc}{\hbox{$#6$}}%
\else
  \savebox{\tempboxa}{\hbox{$\hbox{$\tokena$}\times\hbox{$\tokenb$}$}}%
  \savebox{\tempboxb}{\hbox{$\hbox{$\tokena$}\times\hbox{$\tokend$}$}}%
  \savebox{\tempboxc}{\hbox{$\hbox{$\tokena$}\times\hbox{$#6$}$}}%
\fi
\ypos=\height
\divide\ypos by 2
\xpos=\ypos
\advance\xpos by \width
\bfig
\putCtrianglep<-1`1`1;\ypos>(0,0)[`\tokenc`;#5`#6`{#7}]%
\puthmorphism(\ypos,0)[\tokend`\usebox{\tempboxb}`{#8}]{\width}{-1}b%
\puthmorphism(\ypos,\height)[\tokenb`\usebox{\tempboxa}`{#4}]{\width}{-1}a%
\advance\ypos by \width
\putvmorphism(\ypos,\height)[``\usebox{\tempboxc}]{\height}1r%
\efig
}}
 
\def\recurse{\@ifnextchar <{\recursep}{\recursep<\width`\height>}}
 
\def\puttwohmorphisms(#1,#2)[#3`#4;#5`#6]#7#8#9{{%
\puthmorphism(#1,#2)[#3`#4`]{#7}0a
\ypos=#2
\advance\ypos by 20
\puthmorphism(#1,\ypos)[\phantom{#3}`\phantom{#4}`#5]{#7}{#8}a
\advance\ypos by -40
\puthmorphism(#1,\ypos)[\phantom{#3}`\phantom{#4}`#6]{#7}{#9}b
}}
 
\def\puttwovmorphisms(#1,#2)[#3`#4;#5`#6]#7#8#9{{%
\putvmorphism(#1,#2)[#3`#4`]{#7}0a
\xpos=#1
\advance\xpos by -20
\putvmorphism(\xpos,#2)[\phantom{#3}`\phantom{#4}`#5]{#7}{#8}l
\advance\xpos by 40
\putvmorphism(\xpos,#2)[\phantom{#3}`\phantom{#4}`#6]{#7}{#9}r
}}
 
\def\puthcoequalizer(#1)[#2`#3`#4;#5`#6`#7]#8#9{{%
\setpos(#1)%
\puttwohmorphisms(\xpos,\ypos)[#2`#3;#5`#6]{#8}11%
\advance\xpos by #8
\puthmorphism(\xpos,\ypos)[\phantom{#3}`#4`#7]{#8}1{#9}
}}
 
\def\putvcoequalizer(#1)[#2`#3`#4;#5`#6`#7]#8#9{{%
\setpos(#1)%
\puttwovmorphisms(\xpos,\ypos)[#2`#3;#5`#6]{#8}11%
\advance\ypos by -#8
\putvmorphism(\xpos,\ypos)[\phantom{#3}`#4`#7]{#8}1{#9}
}}
 
\def\putthreehmorphisms(#1)[#2`#3;#4`#5`#6]#7(#8)#9{{%
\setpos(#1) \settypes(#8)
\if a#9 %
     \vertsize{\tempcounta}{#5}%
     \vertsize{\tempcountb}{#6}%
     \ifnum \tempcounta<\tempcountb \tempcounta=\tempcountb \fi
\else
     \vertsize{\tempcounta}{#4}%
     \vertsize{\tempcountb}{#5}%
     \ifnum \tempcounta<\tempcountb \tempcounta=\tempcountb \fi
\fi
\advance \tempcounta by 60
\puthmorphism(\xpos,\ypos)[#2`#3`#5]{#7}{\arrowtypeb}{#9}
\advance\ypos by \tempcounta
\puthmorphism(\xpos,\ypos)[\phantom{#2}`\phantom{#3}`#4]{#7}{\arrowtypea}{#9}
\advance\ypos by -\tempcounta \advance\ypos by -\tempcounta
\puthmorphism(\xpos,\ypos)[\phantom{#2}`\phantom{#3}`#6]{#7}{\arrowtypec}{#9}
}}
 
\def\setarrowtoks[#1`#2`#3`#4`#5`#6]{%
\def\toka{#1}
\def\tokb{#2}
\def\tokc{#3}
\def\tokd{#4}
\def\toke{#5}
\def\tokf{#6}
}
\def\hex{\@ifnextchar <{\hexp}{\hexp<1000`400>}}
\def\hexp<#1`#2>[#3`#4`#5`#6`#7`#8;#9]{%
\setarrowtoks[#9]
\yext=#2 \advance \yext by #2
\xext=#1 \advance\xext by \yext
\bfig
\putCtriangle<-1`0`1;#2>(0,0)[`#5`;\tokb``\tokd]
\xext=#1 \yext=#2 \advance \yext by #2
\putsquare<1`0`0`1;\xext`\yext>(#2,0)[#3`#4`#7`#8;\toka```\tokf]
\advance \xext by #2
\putDtriangle<0`1`-1;#2>(\xext,0)[`#6`;`\tokc`\toke]
\efig
}
\makeatother

\begin{document}

\title{Topos theory and `neo-realist'\\quantum theory}
\author{Andreas D\"{o}ring\footnote{a.doering@imperial.ac.uk}\\ \normalsize Theoretical Physics Group, Blackett Laboratory\\ \normalsize Imperial College, London}
\date{December 24, 2007}
\maketitle

\begin{abstract}
Topos theory, a branch of category theory, has been proposed as
mathematical basis for the formulation of physical theories. In this article,
we give a brief introduction to this approach, emphasising the logical
aspects. Each topos serves as a `mathematical universe' with an internal
logic, which is used to assign truth-values to all propositions about a
physical system. We show in detail how this works for (algebraic) quantum theory.
\end{abstract}

\bigskip

\begin{flushright}
	\textit{``The problem is all inside your head'', she said to me\\
	the answer is easy if you take it logically}\\
	\medskip
	Paul Simon (from \textit{`50 Ways To Leave Your Lover'})
\end{flushright}

\bigskip

\section{Introduction}

The use of topos theory in the foundations of physics and, in particular, the
foundations of quantum theory was suggested by Chris Isham more than 10 years
ago in \cite{Ish97}. Subsequently, these ideas were developed\ in an
application to the Kochen-Specker theorem (with Jeremy Butterfield,
\cite{IB98,IB99,IBH00,IB02}, for conceptual considerations see \cite{IB00}).
In these papers, the use of a multi-valued, contextual logic for quantum
theory was proposed. This logic is given by the internal logic of a certain
topos of presheaves over a category of contexts. Here, contexts typically are
abelian parts of a larger, non-abelian structure. There are several possible
choices for the context category. We will concentrate on algebraic quantum
theory and use the category $\mathcal{V(R)}$ of unital abelian von Neumann
subalgebras of the non-abelian von Neumann algebra of observables $\mathcal{R}$ 
of the quantum system, as first suggested in \cite{IBH00}.

\ 

The use of presheaves over such a category of contexts is motivated by the
very natural construction of the spectral presheaf $\underline{\Sigma}$, which
collects all the Gel'fand spectra of the abelian subalgebras $V\in
\mathcal{V(R)}$ into one larger structure. The Gel'fand spectra can be seen as
`local state spaces', and the spectral presheaf serves as a state space
analogue for quantum theory. Interestingly, as Isham and Butterfield showed,
this presheaf is not like a space: it has no points (in a
category-theoretical sense), and this fact is exactly equivalent to the
Kochen-Specker theorem.

\ 

The topos approach was developed considerably in the series of papers
\cite{DI07a,DI07b,DI07c,DI07d} by Chris Isham and the author. In these papers,
it was shown how topos theory can serve as a new mathematical framework for
the formulation of physical theories. The basic idea of the topos
programme is that by representing the relevant physical structures (states,
physical quantities and propositions about physical quantities) in suitable
topoi, one can achieve a remarkable structural similarity between classical
and quantum physics. Moreover, the topos programme is general enough to allow
for major generalisations. Theories beyond classical and quantum theory are
conceivable. Arguably, this generality will be needed in a future theory of
quantum gravity, which is expected to go well beyond our conventional theories.

\ 

In this paper, we will concentrate on algebraic quantum theory. We briefly
motivate the mathematical constructions and give the main
definitions.\footnote{We suppose that the reader is familiar with the
definitions of a category, functor and natural transformation.} Throughout, we
concentrate on the logical aspects of the theory. We will show in detail how,
given a state, truth-values are assigned to all propositions about a quantum
system. This is independent of any measurement or observer. For that reason,
we say that the topos approach gives a `neo-realist' formulation of quantum theory.

\ 

\subsection{What is a topos?}

It is impossible to give even the briefest introduction to topos theory here.
At the danger of being highly imprecise, we restrict ourselves to mentioning
some aspects of this well-developed mathematical theory and give a number of 
pointers to the literature. The aim merely is to give a very rough idea of 
the structure and internal logic of a topos. In the next subsection, we 
argue that this mathematical structure may be useful in physics.

\ 

There are a number of excellent textbooks on topos theory, and the reader
should consult at least one of them. We found the following books useful:
\cite{LR03,Gol84,MM92,Joh02a,Joh02b,Bell88,LS86}.

\ 

\textbf{Topoi as mathematical universes.} Every (elementary) topos
$\mathcal{E}$ can be seen as a \emph{mathematical universe}.\ As a category, a
topos $\mathcal{E}$ possesses a number of structures that generalise
constructions that are possible in the category $\operatorname*{\mathbf{Set}}$ of sets
and functions.\footnote{More precisely, \emph{small} sets and functions
between them. Small means that we do not have proper classes. One must take
care in these foundational issues to avoid problems like Russell's paradox.}
Namely, in $\operatorname*{\mathbf{Set}}$, we can construct new sets from given ones in
several ways: let $S,T$ be two sets, then we can form the cartesian product
$S\times T$, the disjoint union $S\amalg T$ and the exponential $S^{T}$, the
set of all functions from $T$ to $S$. These constructions turn out to be
fundamental and can all be phrased in an abstract, categorical manner, where
they are called finite limits, colimits and exponentials, respectively. By
definition, a topos $\mathcal{E}$ has all of these. One consequence of the existence 
of finite limits is that each topos has a \emph{terminal object}, denoted by $1$. 
This is characterised by the property that for any object $A$ in the topos 
$\mathcal{E}$, there exists exactly one arrow from $A$ to $1$. In 
$\operatorname*{\mathbf{Set}}$, a one-element set $1=\{*\}$ is terminal.\footnote{Like 
many categorical constructions, the terminal object is fixed only up to 
isomorphism: any two one-element sets are isomorphic, and any of them can serve as 
a terminal object. Nonetheless, one speaks of \emph{the} terminal object.}

\ 

Of course,
$\operatorname*{\mathbf{Set}}$ is a topos, too, and it is precisely the topos which
usually plays the r\^{o}le of our mathematical universe, since we construct
our mathematical objects starting from sets and functions between them. As a
slogan, we have: a topos $\mathcal{E}$ is a category similar to
$\operatorname*{\mathbf{Set}}$. A very nice and gentle introduction to these aspects of
topos theory is the book \cite{LR03}. Other good sources are
\cite{Gol84,McL71}.

\ 

In order to `do mathematics', one must also have a logic, including a
deductive system. Each topos comes equipped with an \emph{internal logic},
which is of \emph{intuitionistic} type. We very briefly sketch the main
characteristics of intuitionistic logic and the mathematical structures in a
topos that realise this logic.

\ 

\textbf{Intuitionistic logic.} Intuitionistic logic is similar to Boolean logic, 
the main difference being that the \emph{law of excluded middle} need not hold. 
In intuitionistic logic, there is \emph{no} axiom%
\begin{equation}
\vdash a\vee\lnot a \tag{$\ast$}%
\end{equation}
like in Boolean logic. Here, $\lnot a$ is the negation of the formula (or
proposition) $a$. The algebraic structures representing intuitionistic logic
are \emph{Heyting algebras}. A Heyting algebra is a pseudocomplemented, 
distributive lattice\footnote{Lattice is meant in the algebraic sense: a 
partially ordered set $L$ such that any two elements $a,b\in L$ have a minimum 
(greatest lower bound) $a\wedge b$ and a maximum (least upper bound) $a\vee b$ 
in $L$. A lattice $L$ is distributive if and only if $a\vee(b\wedge c)=(a\vee
b)\wedge(a\vee c)$ as well as $a\wedge(b\vee c)=(a\wedge b)\vee(a\wedge c)$
hold for all $a,b,c\in L$.} with zero element $0$ and unit element $1$,
representing `totally false' resp. `totally true'. The pseudocomplement is
denoted by $\lnot$, and one has, for all elements $\alpha$ of a Heyting
algebra $H$,%
\[
\alpha\vee\lnot\alpha\leq1,
\]
in contrast to $\alpha\vee\lnot\alpha=1$ in a Boolean algebra. This means that
the disjunction (\textquotedblleft Or\textquotedblright) of a proposition
$\alpha$ and its negation need not be (totally) true in a Heyting algebra.
Equivalently, one has%
\[
\lnot\lnot\alpha\geq\alpha,
\]
in contrast to $\lnot\lnot\alpha=\alpha$ in Boolean algebras.

\ 

Obviously, Boolean logic is a special case of intuitionistic logic. It is
known from Stone's theorem \cite{Sto36} that each Boolean algebra is
isomorphic to an algebra of (clopen, i.e., closed and open) subsets of a
suitable (topological) space.

\  
Let $X$ be a set, and let $P(X)$ be the power 
set of $X$, that is, the set of subsets of $X$. Given a subset $S\in P(X)$, 
one can ask for each point $x\in X$ whether it lies in $S$ or not. This can be
expressed by the \emph{characteristic function} $\chi_{S}:X\rightarrow
\{0,1\}$, which is defined as%
\[
\chi_{S}(x):=\left\{
\begin{tabular}
[c]{ll}%
$1$ & if $x\in S$\\
$0$ & if $x\notin S$%
\end{tabular}
\ \right.
\]
for all $x\in X$. The two-element set $\{0,1\}$ plays the r\^{o}le of a set of
\emph{truth-values} for propositions (of the form \textquotedblleft$x\in
S$\textquotedblright). Clearly, $1$ corresponds to `true', $0$ corresponds to
`false', and there are no other possibilities. This is an argument about sets,
so it takes place in and uses the logic of the topos $\operatorname*{\mathbf{Set}}$ of
sets and functions. $\operatorname*{\mathbf{Set}}$ is a \emph{Boolean topos}, in which
the familiar two-valued logic and the axiom ($\ast$) hold. (This does not
contradict the fact that the internal logic of topoi is intuitionistic, since
Boolean logic is a special case of intuitionistic logic.)

\ 

In an arbitrary topos, there is a special object $\Omega$, called the
\emph{subobject classifier}, that takes the r\^{o}le of the set $\{0,1\}\simeq
\{$false,true$\}$ of truth-values. Let $B$ be an object in the topos, and let
$A$ be a subobject of $B$. This means that there is a monic $A\rightarrow
B$,\footnote{A \emph{monic} is the categorical version of an injective
function. In the topos $\operatorname*{\mathbf{Set}}$, monics exactly are injective
functions.} generalising the inclusion of a subset $S$ into a larger set $X$.
Like in $\operatorname*{\mathbf{Set}}$, we can also characterise $A$ as a subobject of
$B$ by an arrow from $B$ to the subobject classifier $\Omega$. (In
$\operatorname*{\mathbf{Set}}$, this arrow is the characteristic function $\chi
_{S}:X\rightarrow\{0,1\}$.) Intuitively, this characteristic arrow from $B$ to
$\Omega$ tells us how $A$ `lies in' $B$. The textbook definition is:\\

\begin{definition}
In a category $\mathcal{C}$ with finite limits, a \textbf{subobject
classifier} is an object $\Omega$, together with a monic $\operatorname*{true}%
:1\rightarrow\Omega$, such that to every monic $m:A\rightarrow B$ in
$\mathcal{C}$ there is a unique arrow $\chi$ which, with the given monic,
forms a pullback square\\
\begin{center}
\setsqparms[1`2`2`1;700`700]
\square[A`1`B`\Omega;`m`\operatorname{true}`\chi]
\end{center}
\end{definition}

\  

In $\operatorname*{\mathbf{Set}}$, the arrow $\operatorname{true}:1\rightarrow\{0,1\}$ 
is given by $\operatorname{true}(*)=1$. In general, the subobject classifier $\Omega$ 
need not be a set, since it is an object in the topos $\mathcal{E}$, and the objects of 
$\mathcal{E}$ need not be sets. Nonetheless, there is an abstract notion of 
\emph{elements} (or \emph{points}) in category theory that we can use. The elements of 
$\Omega$ are the truth-values available in the internal logic of our topos $\mathcal{E}$, 
just like `false' and `true', the elements of $\{$false,true$\}$, are the
truth-values available in the topos $\operatorname*{\mathbf{Set}}$.

\ 

To understand the abstract notion of elements, let us consider sets for a
moment. Let $1=\{\ast\}$ be a one-element set, the terminal object in 
$\operatorname*{\mathbf{Set}}$. Let $S$ be a set and consider
an arrow $e$ from $1$ to $S$. Clearly, $e(\ast)\in S$ is one element of $S$.
The set of all functions from $1$ to $S$ corresponds exactly to the elements
of $S$. This idea can be generalised to other categories: if there is a
terminal object $1$, then we consider arrows from $1$ to an object $A$ in the 
category as \emph{elements of }$A$. For example, in the definition of the 
subobject classifier the arrow $\operatorname*{true}:1\rightarrow\Omega$ is an 
element of $\Omega$. It may happen that an object $A$ has no elements, i.e., 
there are no arrows $1\rightarrow A$. It is common to consider arrows from 
subobjects $U$ of $A$ to $A$ as \emph{generalised elements}, but we will not 
need this except briefly in subsection \ref{__GenElsAsGenStates}.

\ 

As mentioned, the elements of the subobject classifier, understood as the
arrows $1\rightarrow\Omega$, are the truth-values. Moreover, the set of these
arrows forms a Heyting algebra (see, for example, section 8.3 in
\cite{Gol84}). This is how (the algebraic representation of) intuitionistic
logic manifests itself in a topos. Another, closely related fact is that the
subobjects of any object $A$ in a topos form a Heyting algebra.\\

\subsection{ Topos theory and physics\label{__ToposThAndPhysics}}

A large part of the work on topos theory in physics consists in showing how
states, physical quantities and propositions about physical quantities can be
represented within a suitable topos attached to the system
\cite{DI07a,DI07b,DI07c,DI07d}. The choice of topos will depend on the theory
type (classical, quantum or, in future developments, even something completely
new). Let us consider classical physics for the moment to motivate this.

\ 

\textbf{Realism in classical physics.} In classical physics, one has a space
of states $\mathcal{S}$, and physical quantities $A$ are represented by
real-valued functions $f_{A}:\mathcal{S}\rightarrow\mathbb{R}$.\footnote{We
assume that $f_{A}$ is at least measurable.} A proposition about a physical
quantity $A$ is of the form \textquotedblleft$A\in\Delta$\textquotedblright,
which means \textquotedblleft the physical quantity $A$ has a value in the
(Borel) set $\Delta$\textquotedblright. This proposition is represented by the
inverse image $f_{A}^{-1}(\Delta)\subseteq\mathcal{S}$. In general,
propositions about the physical system correspond to Borel subsets of the
state space $\mathcal{S}$. If we have two propositions \textquotedblleft%
$A\in\Delta_{1}$\textquotedblright, \textquotedblleft$B\in\Delta_{2}%
$\textquotedblright\ and the corresponding subsets $f_{A}^{-1}(\Delta_{1})$,
$f_{B}^{-1}(\Delta_{2})$, then the intersection $f_{A}^{-1}(\Delta_{1})\cap
f_{B}^{-1}(\Delta_{2})$ corresponds to the proposition \textquotedblleft%
$A\in\Delta_{1}$ and $B\in\Delta_{2}$\textquotedblright, the union $f_{A}%
^{-1}(\Delta_{1})\cup f_{B}^{-1}(\Delta_{2})$ corresponds to \textquotedblleft%
$A\in\Delta_{1}$ or $B\in\Delta_{2}$\textquotedblright, and the complement
$\mathcal{S}\backslash f_{A}^{-1}(\Delta_{1})$ corresponds to the negation
\textquotedblleft$A\notin\Delta_{1}$\textquotedblright. Moreover, given a
state $s$, i.e., an element of the state space $\mathcal{S}$, each proposition
is either true or false: if $s$ lies in the subset of $\mathcal{S}$
representing the proposition, then the proposition is true, otherwise it is
false. Every physical quantity $A$ has a value in the state $s$, namely
$f_{A}(s)\in\mathbb{R}$. Thus classical physics is a \emph{realist} theory in
which propositions have truth-values independent of measurements, observers
etc. The logic is Boolean, since classical physics is based on constructions
with sets and functions, i.e., it takes place in the topos
$\operatorname*{\mathbf{Set}}$. We take this as a rule: if we want to describe a
physical system $S$ as a classical system, then the topos $\operatorname*{\mathbf{Set}}%
$ is used. This means no departure from what is ordinarily done, but it
emphasises certain structural and logical aspects of the theory.

\ 

\textbf{Instrumentalism in quantum theory.} In quantum theory, the
mathematical description is very different. Physical quantities $A$ are
represented by self-adjoint operators $\widehat{A}$ on a Hilbert space
$\mathcal{H}$. While $\mathcal{H}$ can be called a space of states, the states
$\psi\in\mathcal{H}$ play a very different r\^{o}le from those in classical
theory. In particular, a state $\psi$ does not assign values to all physical
quantities, only to those for which $\psi$ happens to be an eigenstate. The
spectral theorem shows that propositions \textquotedblleft$A\in\Delta
$\textquotedblright\ are represented by projection operators $\widehat{E}%
[A\in\Delta]$ on Hilbert space. Unless $\psi$ is an eigenstate of $A$, such a
proposition is neither true nor false (except for the trivial cases
$\widehat{E}[A\in\Delta]=\widehat{0}$, which represents trivially false
propositions, and $\widehat{E}[A\in\Delta]=\widehat{1}$, which represents
trivially true propositions). The mathematical formalism of quantum theory is
interpreted in an \emph{instrumentalist} manner: given a state $\psi$, the
proposition \textquotedblleft$A\in\Delta$\textquotedblright\ is assigned a
probability of being true, given by the expectation value $p(A\in\Delta
;\psi):=\left\langle \psi\right\vert \widehat{E}[A\in\Delta]\left\vert
\psi\right\rangle $. This means that upon measurement of the physical quantity
$A$, one will find the measurement result to lie in $\Delta$ with probability 
$p(A\in\Delta;\psi)$. This interpretation depends on measurements and an external
observer. Moreover, the measurement devices (and the observer) are described
in terms of classical physics, not quantum physics.

\ 

\textbf{The motivation from quantum gravity.} An instrumentalist
interpretation cannot describe closed quantum systems, at least there is
nothing much to be said about them from this perspective. A theory of quantum
cosmology or quantum gravity will presumably be a quantum theory of the whole
universe. Since there is no external observer who could perform measurements
in such a theory, instrumentalism becomes meaningless. One of the main
motivations for the topos programme is to overcome or circumvent the usual
instrumentalism of quantum theory and to replace it with a more realist
account of quantum systems. The idea is to use the internal logic of a topos
to assign truth-values to propositions about the system.

\ 

In order to achieve this, we will sketch a new mathematical formulation of
quantum theory that is structurally similar to classical physics. The details
can be found in \cite{DI07a,DI07b,DI07c,DI07d} and references therein.

\ 

\textbf{Plan of the paper.} The starting point is the definition of a formal
language $\mathcal{L}(S)$ attached to a physical system $S$. This is done in
section \ref{_FormalLangForPhys} and emphasises the common structure of
classical and quantum physics. In section \ref{_V(R)AndSetV(R)op}, we
introduce the topos associated to a system $S$ in the case of quantum theory,
and in section \ref{_RepresentingL(S)InSetV(R)op} we briefly discuss the
representation of $\mathcal{L}(S)$ in this topos. The representation of states
and the assignment of truth-values to propositions is treated in section
\ref{_TruthObjectsAndTruthValues}, which is the longest and most detailed
section. Section \ref{_Conclusion} concludes with some remarks on related work
and on possible generalisations.

\section{A formal language for physics\label{_FormalLangForPhys}}

There is a well-developed branch of topos theory that puts emphasis on the
logical aspects. As already mentioned, a topos can be seen as the embodiment
of (higher-order) intuitionistic logic. This point of view is expounded in
detail in Bell's book \cite{Bell88}, which is our standard reference on these
matters. Other excellent sources are \cite{LS86} and part D of \cite{Joh02b}.
The basic concept consists in defining a \emph{formal language} and then
finding a \emph{representation} of it in a suitable topos. As usual in
mathematical logic, the formal language encodes the syntactic aspects of the
theory and the representation provides the semantics. Topoi are a natural
`home'\ for the representation of formal languages encoding intuitionistic
logic, more precisely, \emph{intuitionistic, higher-order, typed predicate
logic with equality}. Typed means that there are several primitive species or
kinds of objects (instead of just sets as primitives), from which sets are
extracted as a subspecies; predicate logic means that one has quantifiers,
namely an existence quantifier $\exists$ (\textquotedblleft it
exists\textquotedblright) and a universal quantifier $\forall$
(\textquotedblleft for all\textquotedblright). Higher-order refers to the fact
that quantification can take place not only over variable individuals, but
also over subsets and functions of individuals as well as iterates of these
constructions. Bell presents a particularly elegant way to specify a formal
language with these properties. He calls this type of language a \emph{local
language}, see chapter 3 of \cite{Bell88}.

\ 

Let $S$ denote a physical system to which we attach a higher-order, typed
language $\mathcal{L}(S)$. We can only sketch the most important aspects here,
details can be found in section 4 of \cite{DI07a}. The language $\mathcal{L}%
(S)$ does not depend on the theory type (classical, quantum, ...), while its
representation of course does. The language contains at least the following
type symbols: $1,\Omega,\Sigma$ and $\mathcal{R}$. The symbol $\Sigma$ serves
as a precursor of the \emph{state object} (see below), the symbol
$\mathcal{R}$ is a precursor of the \emph{quantity-value object}, which is
where physical quantities take their values. Moreover, we require the
existence of \emph{function symbols} of the form $A:\Sigma\rightarrow
\mathcal{R}$. These are the linguistic precursors of physical quantities. For
each type, there exist \emph{variables} of that type. There are a number of
rules how to form terms and formulae (terms of type $\Omega$) from variables
of the various types, including the definition of logical connectives $\wedge$
(\textquotedblleft And\textquotedblright), $\vee$ (\textquotedblleft
Or\textquotedblright) and $\lnot$ (\textquotedblleft Not\textquotedblright).
Moreover, there are axioms giving \emph{rules of inference} that define how to
get new formulae from sets of given formulae. As an example, we mention the
\emph{cut rule}: if $\Gamma$ is a set of formulae and $\alpha$ and $\beta$ are
formulae, then we have%
\[
\frac{\Gamma:\alpha\quad\alpha,\Gamma:\beta}{\Gamma:\beta}%
\]
(here, any free variable in $\alpha$ must be free in $\Gamma$ or $\beta$). In
a representation, where the formulae aquire an interpretation and a `meaning',
this expresses that if $\Gamma$ implies $\alpha$, and $\alpha$ and $\Gamma$
together imply $\beta$, then $\Gamma$ also implies $\beta$. The axioms and
rules of inference are chosen in a way such that the logical operations
satisfy the laws of intuitionistic logic.

\ 

The formal language $\mathcal{L}(S)$ captures a number of abstract properties
of the physical system $S$. For example, if $S$ is the harmonic oscillator,
then we expect to be able to speak about the physical quantity energy in all
theory types, classical or quantum (or other). Thus, among the function
symbols $A:\Sigma\rightarrow\mathcal{R}$, there will be one symbol
$E:\Sigma\rightarrow\mathcal{R}$ which, in a representation, will become the
mathematical entity describing energy. (Which mathematical object that will be
depends on the theory type and thus on the representation.)

\ 

The representation of the language $\mathcal{L}(S)$ takes place in a suitable,
physically motivated topos $\mathcal{E}$. The type symbol $1$ is represented
by the terminal object $1$ in $\mathcal{E}$, the type symbol $\Omega$ is
represented by the subobject classifier $\Omega$. The choice of an appropriate
object $\Sigma$ in the topos that represents the symbol $\Sigma$ depends on
physical insight. The representing object $\Sigma$ is called the \emph{state
object}, and it plays the r\^{o}le of a generalised state space. What actually
is generalised is the space, not the states: $\Sigma$ is an object in a topos
$\mathcal{E}$, which need not be a topos of sets, so $\Sigma$ need not be a
set or space-like. However, as an object in a topos, $\Sigma$ does have
subobjects. These subobjects will be interpreted as (the representatives of)
propositions about the physical quantities, just like in classical physics,
where propositions correspond to subsets of state space. The propositions are
of the form \textquotedblleft$A\in\Delta$\textquotedblright, where $\Delta$
now is a subobject of the object $\mathcal{R}$ that represents the symbol
$\mathcal{R}$. The object $\mathcal{R}$ is called the \emph{quantity-value
object}, and this is where physical quantities take their values. Somewhat
surprisingly, even for ordinary quantum theory this is \emph{not} the real
number object in the topos. Finally, the function symbols $A:\Sigma
\rightarrow\mathcal{R}$ are represented by arrows betwen the objects $\Sigma$
and $\mathcal{R}$ in the topos $\mathcal{E}$.

\ 

In classical physics, the representation is the obvious one: the topos to be
used is the topos $\operatorname*{\mathbf{Set}}$ of sets and mappings, the symbol
$\Sigma$ is represented by a symplectic manifold $\mathcal{S}$ which is the
state space, the symbol $\mathcal{R}$ is represented by the real numbers and
function symbols $A:\Sigma\rightarrow\mathcal{R}$ are represented by
real-valued functions $f_{A}:\mathcal{S}\rightarrow\mathbb{R}$. Propositions
about physical quantities correspond to subsets of the state space.

\section{The context category $\mathcal{V(R)}$ and the topos of presheaves
$\operatorname*{\mathbf{Set}}^{\mathcal{V(R)}^{op}}$\label{_V(R)AndSetV(R)op}}

We will now discuss the representation of $\mathcal{L}(S)$ in the case that
$S$ is to be described as a quantum system. We assume that $S$ is a
non-trivial system that $-$in the usual description$-$ has a Hilbert space
$\mathcal{H}$ of dimension $3$ or greater, and that the physical quantities
belonging to $S$ form a von Neumann algebra $\mathcal{R}(S)\subseteq
\mathcal{B(H)}$ that contains the identity operator $\widehat{1}$.\footnote{
There should arise no confusion between the von Neumann algebra 
$\mathcal{R}=\mathcal{R}(S)$ and the symbol $\mathcal{R}$ of our formal 
language, we hope.}

\ 

From the Kochen-Specker theorem \cite{KS67} we know that there is no state
space model of quantum theory if the algebra of observables is $\mathcal{B(H)}%
$ (for the generalisation to von Neumann algebras see \cite{Doe05}).
More concretely, there is no state space $\mathcal{S}$ such that the physical
quantities are real-valued functions on $\mathcal{S}$. The reason is that if
there existed such a state space $\mathcal{S}$, then each point (i.e., state)
$s\in\mathcal{S}$ would allow to assign values to all physical quantities at
once, simply by evaluating the functions representing the physical quantities
at $s$. One can show that under very mild and natural conditions, this leads
to a mathematical contradiction.

\ 

For an \emph{abelian} von Neumann algebra $V$, there is no such obstacle: the
\emph{Gel'fand spectrum} $\Sigma_{V}$ of $V$ can be interpreted as a state space, and 
the Gel'fand transforms $\overline{A}$ of self-adjoint operators $\widehat{A}\in
V$, representing physical quantities, are real-valued functions on 
$\Sigma_{V}$. The Gel'fand spectrum
$\Sigma_{V}$ of an abelian von Neumann algebra $V$ consists of the pure states
$\lambda$ on $V$ (see e.g. \cite{KR83}). Each $\lambda\in\Sigma_{V}$ also is a multiplicative state;
for all $\widehat{A},\widehat{B}\in V$, we have%
\[
\lambda(\widehat{A}\widehat{B})=\lambda(\widehat{A})\lambda(\widehat{B}),
\]
which, for projections $\widehat{P}\in\mathcal{P}(V)$, implies
\[
\lambda(\widehat{P})=\lambda(\widehat{P}^{2})=\lambda(\widehat{P}%
)\lambda(\widehat{P})\in\{0,1\}.
\]
Finally, each $\lambda\in\Sigma_{V}$ is an algebra homomorphism from $V$ to
$\mathbb{C}$. The Gel'fand spectrum $\Sigma_{V}$ is equipped with the weak* 
topology and thus becomes a compact Hausdorff space.

\ 

Let $\widehat{A}\in V$ and define%
\begin{align*}
\overline{A}:\Sigma_{V}  & \longrightarrow\mathbb{C}\\
\lambda & \longmapsto\overline{A}(\lambda):=\lambda(\widehat{A}).
\end{align*}
The function $\overline{A}$ is called the \emph{Gel'fand transform} of
$\widehat{A}$. It is a continuous function such that
$\operatorname*{im}\overline{A}=\operatorname*{sp}\widehat{A}$. In particular,
if $\widehat{A}$ is self-adjoint, then $\lambda(\widehat{A})\in
\operatorname*{sp}\widehat{A}\subset\mathbb{R}$. The mapping
\begin{align*}
V  & \longrightarrow C(\Sigma_{V})\\
\widehat{A}  & \longmapsto\overline{A}%
\end{align*}
is called the \emph{Gel'fand transformation} on $V$. It is an isometric 
$\ast$-isomorphism between $V$ and $C(\Sigma_{V})$.\footnote{Of course, all this 
holds more generally for abelian $C^{\ast}$-algebras. We concentrate on von Neumann 
algebras, since we need these in our application.}

\ 

This leads to the idea of considering the set $\mathcal{V(R)}$ of
non-trivial unital abelian von Neumann subalgebras of $\mathcal{R}%
$.\footnote{We exclude the trivial algebra $\mathbb{C}\widehat{1}$, which is a
subalgebra of all other subalgebras.} These abelian subalgebras are also
called \emph{contexts}. $\mathcal{V(R)}$ is partially ordered by inclusion and
thus becomes a category. There is an arrow $i_{V^{\prime}V}:V^{\prime
}\rightarrow V$ if and only if $V^{\prime}\subseteq V$, and then
$i_{V^{\prime}V}$ is just the inclusion (or the identity arrow if $V^{\prime
}=V$). The category $\mathcal{V(R)}$ is called the \emph{context category} and
serves as our index category. The process of going from one abelian algebra
$V$ to a smaller algebra $V^{\prime}\subset V$ can be seen as a process of
\emph{coarse-graining}: the algebra $V^{\prime}$ contains less physical
quantities (self-adjoint operators), so we can describe less physics in
$V^{\prime}$ than in $V$. We collect all the `local state spaces' $\Sigma_{V}$ into one large object:

\begin{definition}
The \textbf{spectral presheaf} $\underline{\Sigma}$ is the presheaf\footnote{A
presheaf is a contravariant functor from $\mathcal{V(R)}$ to
$\operatorname*{\mathbf{Set}}$, and obviously, $\underline{\Sigma}$ is of this kind. In
our notation, presheaves will always be underlined.} over $\mathcal{V(R)}$ defined

\begin{enumerate}
\item[a)] on objects: for all $V\in\mathcal{V(R)}$, $\underline{\Sigma}_{V}=\Sigma_{V}$ is the
Gel'fand spectrum of $V$,

\item[b)] on arrows: for all $i_{V^{\prime}V}$, $\underline{\Sigma
}(i_{V^{\prime}V}):\underline{\Sigma}_{V}\rightarrow\underline{\Sigma
}_{V^{\prime}}$ is given by restriction, $\lambda\mapsto\lambda|_{V^{\prime}%
}.$
\end{enumerate}
\end{definition}

The spectral presheaf was first considered by Chris Isham and Jeremy Butterfield 
in the series \cite{IB98,IB99,IBH00,IB02} (see in particular the third of these 
papers). The presheaves over $\mathcal{V(R)}$ form a topos $\operatorname*{\mathbf{Set}}%
^{\mathcal{V(R)}^{op}}$. The arrows in this topos are natural transformations
between the presheaves. Isham and Butterfield developed the idea that this is the 
appropriate topos for quantum theory. The object $\underline{\Sigma}$ in 
$\operatorname*{\mathbf{Set}}^{\mathcal{V(R)}^{op}}$ serves as a state space analogue. 
In the light of the new developments in \cite{DI07a}-\cite{DI07d}, using formal languages, 
we identify $\underline{\Sigma}$ as the state object in 
$\operatorname*{\mathbf{Set}}^{\mathcal{V(R)}^{op}}$, i.e., the representative of the 
symbol $\Sigma$ of our formal language $\mathcal{L}(S)$.

\ 

Isham and Butterfield showed that the Kochen-Specker theorem is exactly
equivalent to the fact that the spectral presheaf $\underline{\Sigma}$ has no
elements, in the sense that there are no arrows from the terminal object
$\underline{1}$ in $\operatorname*{\mathbf{Set}}^{\mathcal{V(R)}^{op}}$ to
$\underline{\Sigma}$. It is not hard to show that having an element of
$\underline{\Sigma}$ would allow the assignment of real values to all physical
quantities at once.

\section{Representing $\mathcal{L}(S)$ in the presheaf topos
$\operatorname*{\mathbf{Set}}^{\mathcal{V(R)}^{op}}$%
\label{_RepresentingL(S)InSetV(R)op}}

\textbf{The quantity-value object for quantum theory.} We already have
identified the topos for the quantum-theoretical description of a system $S$
and the state object $\underline{\Sigma}$ in this topos. Let $V\in
\mathcal{V(R)}$ be a context, then $\downarrow\!\!V:=\{V^{\prime}\in
\mathcal{V(R)}\mid V^{\prime}\subseteq V\}$ denotes the set of all subalgebras of
$V$, equipped with the partial order inherited from $\mathcal{V(R)}$. It can
be shown that the symbol $\mathcal{R}$ should be represented by
the following presheaf \cite{DI07c}:

\begin{definition}
The \textbf{presheaf }$\underline{\mathbb{R}^{\leftrightarrow}}$\textbf{ of
order-preserving and -reversing functions} on $\mathcal{V(R)}$ is defined

\begin{enumerate}
\item[a)] on objects: for all $V\in\mathcal{V(R)}$, $\underline{\mathbb{R}%
^{\leftrightarrow}}_{V}:=\{(\mu,\nu)\mid\mu:\downarrow\!\!V\rightarrow\mathbb{R}$ is
order-preserving, $\nu:\downarrow\!\!V\rightarrow\mathbb{R}$ is order-reversing
and $\mu\leq\nu\},$

\item[b)] on arrows: for all $i_{V^{\prime}V}$, $\underline{\mathbb{R}%
^{\leftrightarrow}}(i_{V^{\prime}V}):\underline{\mathbb{R}^{\leftrightarrow}%
}_{V}\rightarrow\underline{\mathbb{R}^{\leftrightarrow}}_{V^{\prime}}$ is
given by restriction, $(\mu,\nu)\mapsto(\mu|_{V^{\prime}},\nu|_{V^{\prime}})$.
\end{enumerate}
\end{definition}

Here, an order-preserving function $\mu:\downarrow\!\!V\rightarrow\mathbb{R}$ is
a function such that $V''\subseteq V'$ (where $V',V''\in \downarrow\!\!V$) 
implies $\mu(V'')\leq\mu(V')$. Order-reversing functions are defined analogously.

\ 

The presheaf $\underline{\mathbb{R}^{\leftrightarrow}}$ is \emph{not} the
real-number object $\underline{\mathbb{R}}$ in the topos $\operatorname*{\mathbf{Set}}%
^{\mathcal{V(R)}^{op}}$, which is the constant presheaf defined by
$\underline{\mathbb{R}}(V):=\mathbb{R}$ for all $V$ and $\underline
{\mathbb{R}}(i_{V^{\prime}V}):\mathbb{R\rightarrow R}$ as the identity. From
the Kochen-Specker theorem, we would not expect that physical quantities take
their values in the real numbers. (This does not mean that the results of
measurements are not real numbers. We do not discuss measurement here.) More
importantly, the presheaf $\underline{\mathbb{R}^{\leftrightarrow}}$ takes
into account the coarse-graining inherent in the base category $\mathcal{V(R)}%
$: at each stage $V$, a pair $(\mu,\nu)$ consisting of an order-preserving and
an order-reversing function defines a whole range or interval $[\mu
(V),\nu(V)]$ of real numbers, not just a single real number. (It can happen
that $\mu(V)=\nu(V)$.) If we go to a smaller subalgebra $V^{\prime}\subset V$,
which is a kind of coarse-graining, then we have $\mu(V^{\prime})\leq\mu(V)$
and $\nu(V^{\prime})\geq\nu(V)$, so the corresponding interval $[\mu
(V^{\prime}),\nu(V^{\prime})]$ can only become larger.

\ 

\textbf{The representation of function symbols }$A:\Sigma\rightarrow
\mathcal{R}$\textbf{. }In order to represent a physical quantity $A$ belonging
to the system $S$ as an arrow from $\underline{\Sigma}$ to the presheaf
$\underline{\mathbb{R}^{\leftrightarrow}}$ of \ `values', we have to use a
two-step process.

\ 

\textbf{1.} Let $\widehat{A}\in\mathcal{R}$ be the self-adjoint operator
representing $A$. For each abelian subalgebra $V\in\mathcal{V(R)}$, we define%
\begin{align*}
\delta^{o}(\widehat{A})_{V} &  :=\bigwedge\{\widehat{B}\in V_{sa}|\widehat
{B}\geq_{s}\widehat{A}\},\\
\delta^{i}(\widehat{A})_{V} &  :=\bigvee\{\widehat{C}\in V_{sa}|\widehat
{C}\leq_{s}\widehat{A}\}.
\end{align*}
Here, the \emph{spectral order} on self-adjoint operators is used \cite{Ols71,deG05}. 
This is defined for all self-adjoint operators $\widehat{A},\widehat{B}$ with 
spectral families $\widehat{E}^{A}$ resp. $\widehat{E}^{B}$ as%
\[
\widehat{A}\leq_{s}\widehat{B}:\Leftrightarrow(\forall\lambda\in
\mathbb{R}:\widehat{E}_{\lambda}^{A}\geq\widehat{E}_{\lambda}^{B}).
\]
Equipped with the spectral order, the set of self-adjoint operators in a von
Neumann algebra becomes a boundedly complete lattice. In particular, the
mappings $\delta_{V}^{o},\delta_{V}^{i}:\mathcal{R}_{sa}\rightarrow V_{sa}$
are well-defined. We call these mappings \emph{outer} and \emph{inner
daseinisation}, respectively. The outer daseinisation $\delta^{o}(\widehat
{A})_{V}$ of $\widehat{A}$ to the context $V$ is the approximation from above
by the smallest self-adjoint operator in $V$ that is spectrally larger than
$\widehat{A}$. Likewise, the inner daseinisation $\delta^{i}(\widehat{A})_{V}$
is the approximation from below by the largest self-adjoint operator in $V$
that is spectrally smaller then $\widehat{A}$. Since the spectral order is
coarser than the usual, linear order, we have, for all $V$,%
\[
\delta^{i}(\widehat{A})_{V}\leq\widehat{A}\leq\delta^{o}(\widehat{A})_{V}.
\]
One can show that the spectra of $\delta^{i}(\widehat{A})_{V}$ and $\delta
^{o}(\widehat{A})_{V}$ are subsets of the spectrum of $\widehat{A}$, which
seems physically very sensible. If we used the approximation in the linear
order, this would not hold in general. The approximation of self-adjoint
operators in the spectral order was suggested by de Groote \cite{deG05b, 
deG07}. If $V^{\prime}\subset V$, then, by construction,
$\delta^{i}(\widehat{A})_{V^{\prime}}\leq_{s}\delta^{i}(\widehat{A})_{V}$ and
$\delta^{o}(\widehat{A})_{V^{\prime}}\geq_{s}\delta^{o}(\widehat{A})_{V}$,
which implies%
\begin{align*}
\delta^{i}(\widehat{A})_{V^{\prime}}  & \leq\delta^{i}(\widehat{A})_{V},\\
\delta^{o}(\widehat{A})_{V^{\prime}}  & \geq\delta^{o}(\widehat{A})_{V}.
\end{align*}
In this sense, the approximations to $\widehat{A}$ become coarser if the
context becomes smaller.

\ 

\textbf{2.} Now that we have constructed a pair $(\delta^{i}(\widehat{A}%
)_{V},\delta^{o}(\widehat{A})_{V})$ of operators approximating $\widehat{A}$
from below and from above for each context $V$, we can define a natural
transformation $\breve{\delta}(\widehat A)$ from $\underline{\Sigma}$ to
$\underline{\mathbb{R}^{\leftrightarrow}}$ in the following way: let
$V\in\mathcal{V(R)}$ be a context, and let $\lambda\in\underline{\Sigma}_{V}$
be a pure state of $V$. Then define, for all $V^{\prime}\in\downarrow\!\!V$,%
\[
\mu_{\lambda}(V^{\prime}):=\lambda(\delta^{i}(\widehat{A})_{V^{\prime}})=
\overline{\delta^{i}(\widehat{A})_{V^{\prime}}}(\lambda),
\]
where $\overline{\delta^{i}(\widehat{A})_{V^{\prime}}}$ is the Gel'fand 
transform of the self-adjoint operator $\delta^{i}(\widehat{A})_{V^{\prime}}$.
From the theory of abelian $C^{\ast}$-algebras, it is known that
$\lambda(\delta^{i}(\widehat{A})_{V^{\prime}})\in\operatorname*{sp}(\delta
^{i}(\widehat{A})_{V^{\prime}})$ (see e.g. \cite{KR83}). Let $V^{\prime
},V^{\prime\prime}\in\downarrow\!\!V$ such that $V^{\prime\prime}\subset
V^{\prime}$. We saw that $\delta^{i}(\widehat{A})_{V^{\prime\prime}}\leq
\delta^{i}(\widehat{A})_{V^{\prime}}$, which implies $\lambda(\delta
^{i}(\widehat{A})_{V^{\prime\prime}})\leq\lambda(\delta^{i}(\widehat
{A})_{V^{\prime}})$, so $\mu_{\lambda}:\downarrow\!\!V\rightarrow\mathbb{R}$ is
an order-preserving function. Analogously, let%
\[
\nu_{\lambda}(V^{\prime}):=\lambda(\delta^{o}(\widehat{A})_{V^{\prime}})=
\overline{\delta^{o}(\widehat{A})_{V^{\prime}}}(\lambda)
\]
for all $V^{\prime}\in\downarrow\!\!V$. We obtain an order-reversing function
$\nu_{\lambda}:\downarrow\!\!V\rightarrow\mathbb{R}$. Then, for all
$V\in\mathcal{V(R)}$, let%
\begin{align*}
\breve{\delta}(\widehat A)(V):\underline{\Sigma}_{V} &  \longrightarrow
\underline{\mathbb{R}\mathbf{^{\leftrightarrow}}}_{V}\\
\lambda &  \longmapsto(\mu_{\lambda},\nu_{\lambda}).
\end{align*}
By construction, these mappings are the components of a natural transformation
$\breve{\delta}(\widehat A)$$:\underline{\Sigma}\rightarrow\underline
{\mathbb{R}^{\leftrightarrow}}$. For all $V,V'\in\mathcal{V(R)}$ such that 
$V'\subseteq V$, we have a commuting diagram
\begin{center}
\setsqparms[1`1`1`1;1000`700]
\square[\underline{\Sigma}_V`\underline{\Sigma}_{V'}`
\underline{\mathbb{R}\mathbf{^{\leftrightarrow}}}_{V}`
\underline{\mathbb{R}\mathbf{^{\leftrightarrow}}}_{V'};
\underline{\Sigma}(i_{V'V})`\breve{\delta}(\widehat A)(V)`\breve{\delta}(\widehat A)(V')`
\underline{\mathbb{R}\mathbf{^{\leftrightarrow}}}(i_{V'V})]
\end{center}
The arrow $\breve{\delta}(\widehat A):\underline{\Sigma}\rightarrow\underline{\mathbb{R}
\mathbf{^{\leftrightarrow}}}$ in the presheaf topos
$\operatorname*{\mathbf{Set}}^{\mathcal{V(R)}^{op}}$ is the representative of the
physical quantity $A$, which is abstractly described by the function symbol
$A:\Sigma\rightarrow\mathcal{R}$ in our formal language. The physical content,
namely the appropriate choice of the self-adjoint operator $\widehat{A}$ from
which we construct the arrow $\breve{\delta}(\widehat A)$, is not part of the
language, but part of the representation.\footnote{The current scheme is not
completely topos-internal yet. It is an open question if every arrow from
$\underline{\Sigma}$ to $\underline{\mathbb{R}^{\leftrightarrow}}$ comes from
a self-adjoint operator. This is why we start from a self-adjoint operator
$\widehat{A}$ to construct $\breve{\delta}(\widehat A)$. We are working on a
more internal characterisation.}

\ 

\textbf{The representation of propositions.} As discussed in subsection
\ref{__ToposThAndPhysics}, in classical physics the subset of state space
$\mathcal{S}$ representing a proposition \textquotedblleft$A\in\Delta
$\textquotedblright\ is constructed by taking the inverse image $f_{A}%
^{-1}(\Delta)$ of $\Delta$ under the function representing $A$. We will use
the analogous construction in the topos formulation of quantum theory: the set
$\Delta$ is a subset (that is, subobject) of the quantity-value object
$\mathbb{R}$ in classical physics, so we start from a subobject $\Theta$ of
the presheaf $\underline{\mathbb{R}^{\leftrightarrow}}$. We get a subobject of
the state object $\underline{\Sigma}$ by pullback along
$\breve{\delta}(\widehat A)$, which we denote by
$\breve{\delta}(\widehat A)$$^{-1}(\Theta)$.\footnote{This is a well-defined
categorical construction, since the pullback of a monic is a monic, so we get
a subobject of $\underline{\Sigma}$ from a subobject of $\underline
{\mathbb{R}^{\leftrightarrow}}$.} For details see subsection 3.6 in
\cite{DI07c} and also \cite{HS07}.

\ 

In both classical and quantum theory, propositions are represented by
subobjects of the quantity-value object (state space $\mathcal{S}$ resp.
spectral presheaf $\underline{\Sigma}$). Such subobjects are constructed by
pullback from subobjects of the quantity-value object (real numbers
$\mathbb{R}$ resp. presheaf of order-preserving and -reversing functions
$\underline{\mathbb{R}^{\leftrightarrow}}$). The interpretation and meaning of
such propositions is determined by the internal logic of the topos
($\operatorname*{\mathbf{Set}}$ resp. $\operatorname*{\mathbf{Set}}^{\mathcal{V(R)}^{op}}$). In
the classical case, where $\operatorname*{\mathbf{Set}}$ is used, this is the ordinary
Boolean logic that we are familiar with. In the quantum case, the internal
logic of the presheaf topos $\operatorname*{\mathbf{Set}}^{\mathcal{V(R)}^{op}}$ has to
be used. This intuitionistic logic can be interpreted using Kripke-Joyal
semantics, see e.g. chapter VI in \cite{MM92}.

\ 

\textbf{The Heyting algebra structure of subobjects.} In the next section, we
discuss the representation of states in the topos $\operatorname*{\mathbf{Set}}%
^{\mathcal{V(R)}^{op}}$ and the assignment of truth-values to propositions.
Before doing so, it is worth noting that the subobjects of $\underline{\Sigma
}$ form a Heyting algebra (since the subobjects of any object in a topos do),
so we have mapped propositions \textquotedblleft$A\in\Delta$\textquotedblright%
\ (understood as discussed) to a \emph{distributive} lattice with a
pseudocomplement. Together with the results from the next section, we have a
completely new form of quantum logic, based upon the internal logic of the
presheaf topos $\operatorname*{\mathbf{Set}}^{\mathcal{V(R)}^{op}}$. Since this is a
distributive logic and since the internal logic of a topos has powerful rules
of inference, this kind of quantum logic is potentially much better
interpretable than ordinary quantum logic of the Birkhoff-von Neumann kind.
The latter type of quantum logic and its generalisations are based on
nondistributive structures and lack a deductive system.

\section{Truth objects and truth-values\label{_TruthObjectsAndTruthValues}}

In classical physics, a state is just a point of state space.\footnote{One
might call this a \emph{pure} state, though this is not customary in classical
physics. Such a state actually is a point measure on state space, in contrast
to more general probability measures that describe general states. We only
consider pure states here and identify the point measure with the corresonding
point of state space.} Since, as we saw, the spectral presheaf $\underline
{\Sigma}$ has no elements (or, global elements\footnote{Elements
$\underline{1}\rightarrow\underline{\mathcal{P}}$ of a presheaf $\underline
{\mathcal{P}}$ are called \emph{global elements} or \emph{global sections} in
category theory. We follow this convention to avoid confusion with points or
elements of sets.}), we must represent states differently in the presheaf
topos $\operatorname*{\mathbf{Set}}^{\mathcal{V(R)}^{op}}$.\\

\subsection{Generalised elements as generalised
states\label{__GenElsAsGenStates}}

One direct way, suggested in \cite{HS07}, is the following generalisation:
$\underline{\Sigma}$ has no global elements $\underline{1}\rightarrow
\underline{\Sigma}$, but it does have subobjects $\underline{U}\hookrightarrow
\underline{\Sigma}$. In algebraic geometry and more generally in category
theory, such monics (and, more generally, arbitrary arrows) are called 
generalised elements \cite{LR03}. We could postulate that these subobjects, 
or some of them, are `generalised states'. Consider another subobject of 
$\underline{\Sigma}$ that represents a proposition
\textquotedblleft$A\in\Delta$\textquotedblright\ about the quantum system,
given by its characteristic arrow $\chi_{\underline{S}}:\underline{\Sigma
}\rightarrow\underline{\Omega}$. Then we can compose these arrows%
\[
\underline{U}\hookrightarrow\underline{\Sigma}\rightarrow\underline{\Omega}%
\]
to obtain an arrow $\underline{U}\rightarrow\underline{\Omega}$. This is
\emph{not} a global element $\underline{1}\rightarrow\underline{\Omega}$ of
$\underline{\Omega}$, and by construction, it cannot be, since $\underline
{\Sigma}$ has no global elements, but it is a generalised element of
$\underline{\Omega}$. It might be possible to give a physical meaning to these
arrows $\underline{U}\rightarrow\underline{\Omega}$ if one can (a) give
physical meaning to the subobject $\underline{U}\hookrightarrow\underline
{\Sigma}$, making clear what a generalised state actually is, and (b) give a
logical and physical interpretation of an arrow $\underline{U}\rightarrow
\underline{\Omega}$. While a global element $\underline{1}\rightarrow
\underline{\Omega}$ is interpreted as a truth-value in the internal logic of a
topos, the logical interpretation of an arrow $\underline{U}%
\rightarrow\underline{\Omega}$ is not so clear.

\ 

We want to emphasise that mathematically, the above construction is perfectly
well-defined. It remains to be worked out if a physical and logical meaning
can be attached to it.\\

\subsection{The construction of truth objects}

We now turn to the construction of so-called `truth objects' from pure quantum
states $\psi$, see also \cite{DI07b}. (To be precise, a unit vector $\psi$ in the
Hilbert space $\mathcal{H}$ represents a vector state $\varphi_{\psi
}:\mathcal{R}\rightarrow\mathbb{C}$ on a von Neumann algebra, given by
$\varphi_{\psi}(\widehat{A}):=\left\langle \psi\right\vert \widehat
{A}\left\vert \psi\right\rangle $ for all $\widehat{A}\in\mathcal{R}$. If
$\mathcal{R}=\mathcal{B(H)}$, then every $\varphi_{\psi}$ is a pure state.) Of
course, the Hilbert space $\mathcal{H}$ is the Hilbert space on which the von
Neumann algebra of observables $\mathcal{R}\subseteq\mathcal{B(H)}$ is
represented. This is the most direct way in which Hilbert space actually
enters the mathematical constructions inside the topos $\operatorname*{\mathbf{Set}}%
^{\mathcal{V(R)}^{op}}$. However, we will see how this direct appeal to 
Hilbert space possibly can be circumvented.

\ 

Given a subobject of $\underline{\Sigma}$ that represents some proposition, a
truth object will allow us to construct a global element $\underline
{1}\rightarrow\underline{\Sigma}$ of $\underline{\Sigma}$, as we will show in
subsection \ref{__AssignmentOfTruthValues}. This means that from a proposition
and a state, we \emph{do} get an actual truth-value for that proposition in
the internal logic of the topos $\operatorname*{\mathbf{Set}}^{\mathcal{V(R)}^{op}}$.
The construction of truth objects is a direct generalisation of the classical case.

\ 

For the moment, let us consider sets. Let $S$ be a subset of some larger set
$X$, and let $x\in X$. Then%
\[
(x\in S)\Leftrightarrow(S\in U(x)),
\]
where $U(x)$ denotes the set of neighbourhoods of $x$ in $X$. The key
observation is that while the l.h.s. cannot be generalised to the topos
setting, since we cannot talk about points like $x$, the r.h.s. can. The task
is to define neighbourhoods in a suitable manner. We observe that $U(x)$ is a
subset of the power set $PX=P(X)$, which is the same as an element of the
power set of the power set $PPX=P(P(X))$.

\ 

This leads to the idea that for each context $V\in\mathcal{V(R)}$, we must
choose an appropriate set of subsets of the Gel'fand spectrum $\underline
{\Sigma}_{V}$ such that these sets of subsets form an element in
$PP\underline{\Sigma}$. Additionally, the subsets we choose at each stage $V$
should be clopen, since the clopen subsets $P_{cl}(\underline{\Sigma}_{V})$
form a lattice that is isomorphic to the lattice $\mathcal{P}(V)$ of
projections in $V$.

\ 

The main difficulty lies in the fact that the spectral presheaf $\underline
{\Sigma}$ has no global elements, which is equivalent to the Kochen-Specker
theorem. A global element, if it existed, would pick one point
$\lambda_{V}$ from each Gel'fand spectrum $\underline{\Sigma}_{V}$
($V\in\mathcal{V(R)}$) such that, whenever $V^{\prime}\subset V$, we would
have $\lambda_{V^{\prime}}=\lambda_{V}|_{V^{\prime}}$. If we had such global
elements, we could define neighbourhoods for them by taking, for each
$V\in\mathcal{V(R)}$, neighbourhoods of $\lambda_{V}$ in $\underline{\Sigma
}_{V}$.

\ 

Since no such global elements exist, we cannot expect to have neighbourhoods
of \emph{points} at each stage. Rather, we will get neighbourhoods of
\emph{sets} at each stage $V$, and only for particular $V$, these sets will
have just one element. In any case, the sets will depend on the state $\psi$
in a straighforward manner. We define:

\begin{definition}
Let $\psi\in\mathcal{H}$ be a unit vector, let $\widehat{P}_{\psi}$ the
projection onto the corresponding one-dimensional subspace (i.e., ray) of
$\mathcal{H}$, and let $P_{cl}(\underline{\Sigma}_{V})$ be the clopen subsets
of the Gel'fand spectrum $\underline{\Sigma}_{V}$. If $S\in P_{cl}%
(\underline{\Sigma}_{V})$, then $\widehat{P}_{S}\in\mathcal{P}(V)$ denotes the
corresponding projection. The truth object $\mathbb{T}^{\psi}=(\mathbb{T}%
_{V}^{\psi})_{V\in\mathcal{V(R)}}$ is given by%
\[
\forall V\in\mathcal{V(R)}:\mathbb{T}_{V}^{\psi}:=\{S\in P_{cl}(\underline
{\Sigma}_{V})\mid\left\langle \psi\right\vert \widehat{P}_{S}\left\vert
\psi\right\rangle =1\}.
\]

\end{definition}

At each stage $V$, $\mathbb{T}_{V}^{\psi}$ collects all subsets $S$ of
$\underline{\Sigma}_{V}$ such that the expectation value of the projection
corresponding to this subset is $1$. From this definition, it is not clear at
first sight that the set $\mathbb{T}_{V}^{\psi}$ can be seen as a set of neighbourhoods.

\begin{lemma}
We have the following equalities:%
\begin{align*}
\forall V\in\mathcal{V(R)}:\mathbb{T}_{V}^{\psi}  &  =\{S\in P_{cl}%
(\underline{\Sigma}_{V})\mid\left\langle \psi\right\vert \widehat{P}%
_{S}\left\vert \psi\right\rangle =1\}\\
&  =\{S\in P_{cl}(\underline{\Sigma}_{V})\mid\widehat{P}_{S}\geq\widehat
{P}_{\psi}\}\\
&  =\{S\in P_{cl}(\underline{\Sigma}_{V})\mid\widehat{P}_{S}\geq\delta
^{o}(\widehat{P}_{\psi})_{V}\}\\
&  =\{S\in P_{cl}(\underline{\Sigma}_{V})\mid S\supseteq S_{\delta
^{o}(\widehat{P}_{\psi})_{V}}\}.
\end{align*}

\end{lemma}

\begin{proof}
If $\left\langle \psi\right\vert \widehat{P}_{S}\left\vert \psi\right\rangle
=1$, then $\psi$ lies entirely in the subspace of Hilbert space that
$\widehat{P}_{S}$ projects onto. This is equivalent to $\widehat{P}_{S}%
\geq\widehat{P}_{\psi}$. Since $\widehat{P}_{S}\in\mathcal{P}(V)$ and
$\delta^{o}(\widehat{P}_{\psi})_{V}$ is the \emph{smallest} projection in $V$
that is larger than $\widehat{P}_{\psi}$,\footnote{On projections, the
spectral order $\leq_{s}$ and the linear order $\leq$ coincide.} we also have
$\widehat{P}_{S}\geq\delta^{o}(\widehat{P}_{\psi})_{V}$. In the last step, we
simply go from the projections in $V$ to the corresponding clopen subsets of
$\underline{\Sigma}_{V}$.
\end{proof}

\ 

This reformulation shows that $\mathbb{T}_{V}^{\psi}$ actually consists of
subsets of the Gel'fand spectrum $\underline{\Sigma}_{V}$ that can be seen
as some kind of neigbourhoods, not of a single point of $\underline{\Sigma
}_{V}$, but of a certain subset of $\underline{\Sigma}_{V}$, namely
$S_{\delta^{o}(\widehat{P}_{\psi})_{V}}$. In the simplest case, we have
$\widehat{P}_{\psi}\in\mathcal{P}(V)$, so $\delta^{o}(\widehat{P}_{\psi}%
)_{V}=\widehat{P}_{\psi}$. Then $S_{\delta^{o}(\widehat{P}_{\psi})_{V}%
}=S_{\widehat{P}_{\psi}}$, and this subset contains a single element, namely
the pure state $\lambda$ such that%
\[
\lambda(\widehat{P}_{\psi})=1
\]
and $\lambda(\widehat{Q})=0$ for all $\widehat{Q}\in\mathcal{P}(V)$ such that
$\widehat{Q}\widehat{P}_{\psi}=0$. In this case, $\mathbb{T}_{V}^{\psi}$
actually consists of all the clopen neighbourhoods of the point $\lambda$ in
$\underline{\Sigma}_{V}$.

\ 

In general, if $\widehat{P}_{\psi}$ does not lie in the projections
$\mathcal{P}(V)$, then there is no subset of $\underline{\Sigma}_{V}$ that
corresponds directly to $\widehat{P}_{\psi}$. We must first approximate
$\widehat{P}_{\psi}$ by a projection in $V$, and $\delta^{o}(\widehat{P}%
_{\psi})_{V}$ is the smallest projection in $V$ larger than $\widehat{P}%
_{\psi}$. The projection $\delta^{o}(\widehat{P}_{\psi})_{V}$ corresponds to a
subset $S_{\delta^{o}(\widehat{P}_{\psi})_{V}}\subseteq\underline{\Sigma}_{V}$
that may contain more than one element. However, $\mathbb{T}_{V}^{\psi}$ can
still be seen as a set of neighbourhoods, but now of this set $S_{\delta
^{o}(\widehat{P}_{\psi})_{V}}$ rather than of a single point.

\ 

It is an interesting and non-trivial point that the (outer) daseinisation
$\delta^{o}(\widehat{P}_{\psi})_{V}$ ($V\in\mathcal{V(R)}$) shows up in this
construction. We did not discuss this here, but the subobjects of
$\underline{\Sigma}$ constructed from the outer daseinisation of projections
play a central r\^{o}le in the representation of a certain propositional
language $\mathcal{PL}(S)$ that one can attach to a physical system $S$
\cite{DI07a,DI07b}. Moreover, these subobjects are `optimal' in the sense
that, whenever $V^{\prime}\subset V$, the restriction from $S_{\delta
^{o}(\widehat{P})_{V}}$ to $S_{\delta^{o}(\widehat{P})_{V^{\prime}}}$ is
\emph{surjective}, see Theorem 3.1 in \cite{DI07b}. This property can also 
lead the way to a more internal characterisation of truth-objects, without 
reference to a state $\psi$ and hence to Hilbert space.\\

\subsection{Truth objects and Birkhoff-von Neumann quantum logic}

There is yet another point of view on what a truth object $\mathbb{T}^{\psi}$
is, closer the ordinary quantum logic, which goes back to the famous paper
\cite{BV36} by Birkhoff and von Neumann. For now, let us assume that $\mathcal{R}%
=\mathcal{B(H)}$, then we write $\mathcal{P(H)}:=\mathcal{P(B(H))}$ for the
lattice of projections on Hilbert space. In their paper, Birkhoff and von
Neumann identify a proposition \textquotedblleft$A\in\Delta$\textquotedblright%
\ about a quantum system with a projection operator $\widehat{E}[A\in
\Delta]\in\mathcal{P(H)}$ via the spectral theorem \cite{KR83} and interpret the
lattice structure of $\mathcal{P(H)}$ as giving a quantum logic. This is very
different from the topos form of quantum logic, since $\mathcal{P(H)}$ is a
\emph{non-distributive} lattice, leading to all the well-known
interpretational difficulties. Nonetheless, in this subsection we want to
interpret truth objects from the perspective of Birkhoff-von Neumann quantum logic.

\ 

The implication in ordinary quantum logic is given by the partial order on
$\mathcal{P(H)}$: a proposition \textquotedblleft$A\in\Delta_{1}%
$\textquotedblright\ implies a proposition \textquotedblleft$B\in\Delta_{2}%
$\textquotedblright\ (where we can have $B=A$) if and only if $\widehat
{E}[A\in\Delta_{1}]\leq\widehat{E}[B\in\Delta_{2}]$ holds for the
corresponding projections.

\ 

The idea now is that, given a pure state $\psi$ and the corresponding
projection $\widehat{P}_{\psi}$ onto a ray, we can collect all the projections
larger than or equal to $\widehat{P}_{\psi}$. We denote this by%
\[
T^{\psi}:=\{\widehat{P}\in\mathcal{P(H)}\mid\widehat{P}\geq\widehat{P}_{\psi}\}.
\]
The propositions represented by these projections are exactly those
propositions about the quantum system that are (totally) true if the system is
in the state $\psi$. Totally true means `true with probability $1$' in an
instrumentalist interpretation. If, for example, a projection $\widehat
{E}[A\in\Delta]$ is larger than $\widehat{P}_{\psi}$ and hence contained in
$T^{\psi}$, then, upon measurement of the physical quantity $A$, we will find
the measurement result to lie in the set $\Delta$ with certainty (i.e., with 
probability $1$).

\ 

$T^{\psi}$ is a maximal (proper) filter in $\mathcal{P(H)}$. Every pure state 
$\psi$ gives rise to such a maximal filter $T^{\psi}$, and clearly, the mapping 
$\psi\mapsto T^{\psi}$ is injective. We can obtain the truth object 
$\mathbb{T}^{\psi}$ from the maximal filter $T^{\psi}$ simply by defining%
\[
\forall V\in\mathcal{V(R)}:\mathbb{T}_{V}^{\psi}:=T^{\psi}\cap V.
\]
In each context $V$, we collect all the projections larger than $\widehat
{P}_{\psi}$. On the level of propositions, we have all the propositions about
physical quantities $A$ \emph{in the context }$V$ that are totally true in the
state $\psi$.\\

\subsection{The assignment of truth-values to
propositions\label{__AssignmentOfTruthValues}}

We return to the consideration of the internal logic of the topos 
$\operatorname*{\mathbf{Set}}^{\mathcal{V(R)}^{op}}$ and show how to define a 
global element $\underline{1}\rightarrow\underline{\Omega}$ of the subobject 
classifier from a clopen subobject $\underline{S}$ of $\underline{\Sigma}$ 
and a truth object $\mathbb{T}^{\psi}$. The subobject $\underline{S}$
represents a proposition about the quantum
system, the truth object $\mathbb{T}^{\psi}$ represents a state, and
the global element of $\underline{\Omega}$ will be interpreted as the
truth-value of the proposition in the given state. Thus, we make use of the
internal logic of the topos $\operatorname*{\mathbf{Set}}^{\mathcal{V(R)}^{op}}$ of
presheaves over the context category $\mathcal{V(R)}$ to assign truth-values
to all propositions about a quantum system.

\ 

It is well known that the subobject classifier $\underline{\Omega}$ in a topos
of presheaves is the presheaf of \emph{sieves} (see e.g. \cite{MM92}). A sieve
$\sigma$ on an object $A$ in some category $\mathcal{C}$ is a collection of 
arrows with codomain $A$ with the following property: if $f:B\rightarrow A$ is in $\sigma$ and 
$g:C\rightarrow B$ is another arrow in $\mathcal{C}$, then $f\circ g:C\rightarrow A$
is in $\sigma$, too. In other words, a sieve on $A$ is a downward closed set of 
arrows with codomain $A$. Since the context category $\mathcal{V(R)}$ is a partially
ordered set, things become very simple: the only arrows with codomain $V$ are
the inclusions $i_{V^{\prime}V}$. Since such an arrow is specified uniquely by
its domain $V^{\prime}$, we can think of the sieve $\sigma$ on $V$ as consisting of
certain subalgebras $V^{\prime}$ of $V$. If $V^{\prime}\in\sigma$ and
$V^{\prime\prime}\subset V^{\prime}$, then $V^{\prime\prime}\in\sigma$.

\ 

The restriction mappings of the presheaf $\underline{\Omega}$ are given by
pullbacks of sieves. The pullback of sieves over a partially ordered set takes 
a particularly simple form:

\begin{lemma}
\label{L_PullbackOfSievesOnPoset}If $\sigma$ is a sieve on $V\in
\mathcal{V(R)}$ and $V^{\prime}\subset V$, then the pullback $\sigma\cdot
i_{V^{\prime}V}$ is given by $\sigma\cap\downarrow\!\!V^{\prime}$. (This holds
analogously for sieves on any partially ordered set, not just $\mathcal{V(R)}$).
\end{lemma}

\begin{proof}
For the moment, we switch to the arrows notation. By definition, the
pullback $\sigma\cdot i_{V^{\prime}V}$ is given by%
\[
\sigma\cdot i_{V^{\prime}V}:=\{i_{V^{\prime\prime}V^{\prime}}\mid
i_{V^{\prime}V}\circ i_{V^{\prime\prime}V^{\prime}}\in\sigma\}.
\]
We now identify arrows and subalgebras as usual and obtain (using the fact
that $V^{\prime\prime}\subseteq V^{\prime}$ implies $V^{\prime\prime}\subset
V$)%
\[
\{i_{V^{\prime\prime}V^{\prime}}\mid i_{V^{\prime}V}\circ i_{V^{\prime\prime
}V^{\prime}}\in\sigma\}\simeq\{V^{\prime\prime}\subseteq V^{\prime}\mid
V^{\prime\prime}\in\sigma\}=\downarrow\!\!V^{\prime}\cap\sigma.
\]
Since $\downarrow V^{\prime}$ is the maximal sieve on $V^{\prime}$, the
pullback $\sigma\cdot i_{V^{\prime}V}$ is given as the intersection of
$\sigma$ with the maximal sieve on $V^{\prime}$.
\end{proof}

\ 

The \emph{name} $\ulcorner\underline{S}\urcorner$ of the subobject
$\underline{S}$ is the unique arrow $\underline{1}\rightarrow P\underline
{\Sigma}=\underline{\Omega}^{\underline{\Sigma}}$ into the power object of
$\underline{\Sigma}$ (i.e., the subobjects of $\underline{\Sigma}$) that
`picks out' $\underline{S}$ among all subobjects. $\ulcorner\underline
{S}\urcorner$ is a global element of $P\underline{\Sigma}$. Here, one uses the
fact that power objects behave like sets, in particular, they have global
elements. Since we assume that $\underline{S}$ is a \emph{clopen} subobject,
we also get an arrow $\underline{1}\rightarrow P_{cl}\underline{\Sigma}$ into
the clopen power object of $\underline{\Sigma}$, see \cite{DI07b}. We denote this arrow
by $\ulcorner\underline{S}\urcorner$, as well.

\ 

Since $\mathbb{T}^{\psi}\in PP_{cl}\underline{\Sigma}$ is a collection of
clopen subobjects of $\underline{\Sigma}$, it makes sense to ask if
$\underline{S}$ is among them; an expression like $\ulcorner\underline
{S}\urcorner\in\mathbb{T}^{\psi}$ is well-defined. We define, for all 
$V\in\mathcal{V(R)}$, the \emph{valuation}%
\[
v(\ulcorner\underline{S}\urcorner\in\mathbb{T}^{\psi})_{V}:=\{V^{\prime
}\subseteq V\mid\underline{S}(V^{\prime})\in\mathbb{T}_{V^{\prime}}^{\psi}\}.
\]
At each stage $V$, we collect all those subalgebras of $V$ such that
$\underline{S}(V^{\prime})$ is contained in $\mathbb{T}_{V^{\prime}}^{\psi}$.

\ 

In order to construct a global element of the presheaf of sieves $\underline{\Omega}$, 
we must first show that $v(\ulcorner\underline{S}\urcorner\in\mathbb{T}^{\psi})_{V}$ 
is a sieve on $V$. In the proof we use the fact that the subobjects obtained from
daseinisation are optimal in a certain sense.

\begin{proposition}
\label{P_val(S in T)_VIsSieveOnV}$v(\ulcorner\underline{S}\urcorner
\in\mathbb{T}^{\psi})_{V}:=\{V^{\prime}\subseteq V\mid\underline{S}(V^{\prime
})\in\mathbb{T}_{V^{\prime}}^{\psi}\}$ is a sieve on $V$.
\end{proposition}

\begin{proof}
As usual, we identify an inclusion morphism $i_{V^{\prime}V}$ with $V^{\prime
}$ itself, so a sieve on $V$ consists of certain subalgebras of $V$. We have
to show that if $V^{\prime}\in v(\ulcorner\underline{S}\urcorner\in
\mathbb{T}^{\psi})_{V}$ and $V^{\prime\prime}\subset V^{\prime}$, then
$V^{\prime\prime}\in v(\ulcorner\underline{S}\urcorner\in\mathbb{T}^{\psi
})_{V}$. Now, $V^{\prime}\in v(\ulcorner\underline{S}\urcorner\in
\mathbb{T}^{\psi})_{V}$ means that $\underline{S}(V^{\prime})\in
\mathbb{T}_{V^{\prime}}^{\psi}$, which is eqivalent to $\underline
{S}(V^{\prime})\supseteq\underline{S}_{\delta^{o}(\widehat{P}_{\psi})_{V^{\prime}}%
}$. Here, $\underline{S}_{\delta^{o}(\widehat{P}_{\psi})_{V^{\prime}}}$ is the
component at $V^{\prime}$ of the sub-object $\underline{S}_{\delta
^{o}(\widehat{P}_{\psi})}=(\underline{S}_{\delta^{o}(\widehat{P}_{\psi})_{V}%
})_{V\in\mathcal{V(R)}}$ of $\underline{\Sigma}$ obtained from daseinisation
of $\widehat{P}_{\psi}$. According to Thm. 3.1 in \cite{DI07b}, the sub-object
$\underline{S}_{\delta^{o}(\widehat{P}_{\psi})}$ is optimal in the following
sense: when restricting from $V^{\prime}$ to $V^{\prime\prime}$, we have
$\underline{\Sigma}(i_{V^{\prime\prime}V^{\prime}})(\underline{S}_{\delta
^{o}(\widehat{P}_{\psi})_{V^{\prime}}})=\underline{S}_{\delta^{o}(\widehat
{P}_{\psi})_{V^{\prime\prime}}}$, i.e., the restriction is surjective. By
assumption, $\underline{S}(V^{\prime})\supseteq\underline{S}_{\delta^{o}%
(\widehat{P}_{\psi})_{V^{\prime}}}$, which implies%
\[
\underline{S}(V^{\prime\prime})\supseteq\underline{\Sigma}(i_{V^{\prime\prime
}V^{\prime}})(\underline{S}(V^{\prime}))\supseteq\underline{\Sigma
}(i_{V^{\prime\prime}V^{\prime}})(\underline{S}_{\delta^{o}(\widehat{P}_{\psi
})_{V^{\prime}}})=\underline{S}_{\delta^{o}(\widehat{P}_{\psi})_{V^{\prime
\prime}}}.
\]
This shows that $\underline{S}(V^{\prime\prime})\in\mathbb{T}_{V^{\prime\prime}}^{\psi}$ and
hence $V^{\prime\prime}\in v(\ulcorner\underline{S}\urcorner\in\mathbb{T}%
^{\psi})_{V}$.
\end{proof}

\ 

Finally, we have to show that the sieves $v(\ulcorner\underline{S}\urcorner
\in\mathbb{T}^{\psi})_{V}$, $V\in\mathcal{V(R)}$, actually form a global
element of $\underline{\Omega}$, i.e., they all fit together under the 
restriction mappings of the presheaf $\underline{\Omega}$:

\begin{proposition}
The sieves $v(\ulcorner\underline{S}\urcorner\in\mathbb{T}^{\psi})_{V}$,
$V\in\mathcal{V(R)}$, (see Prop. \ref{P_val(S in T)_VIsSieveOnV}) form a
global element of $\underline{\Omega}$.
\end{proposition}

\begin{proof}
From Lemma \ref{L_PullbackOfSievesOnPoset}, is suffices to show that, whenever
$V^{\prime}\subset V$, we have $v(\ulcorner\underline{S}\urcorner\in
\mathbb{T}^{\psi})_{V^{\prime}}=v(\ulcorner\underline{S}\urcorner\in
\mathbb{T}^{\psi})_{V}\cap\downarrow\!\!V^{\prime}$. If $V^{\prime\prime}\in
v(\ulcorner\underline{S}\urcorner\in\mathbb{T}^{\psi})_{V^{\prime}}$, then
$\underline{S}(V^{\prime\prime})\in\mathbb{T}_{V^{\prime\prime}}^{\psi}$,
which implies $V^{\prime\prime}\in v(\ulcorner\underline{S}\urcorner
\in\mathbb{T}^{\psi})_{V}$. Conversely, if $V^{\prime\prime}\in
\downarrow\!\!V^{\prime}$ and $V^{\prime\prime}\in v(\ulcorner\underline{S}\urcorner
\in\mathbb{T}^{\psi})_{V}$, then, again, $\underline{S}(V^{\prime\prime}%
)\in\mathbb{T}_{V^{\prime\prime}}^{\psi}$, which implies $V^{\prime\prime}\in
v(\ulcorner\underline{S}\urcorner\in\mathbb{T}^{\psi})_{V^{\prime}}$.
\end{proof}

\ 

The global element $v(\ulcorner\underline{S}\urcorner\in\mathbb{T}^{\psi
})=(v(\ulcorner\underline{S}\urcorner\in\mathbb{T}^{\psi})_{V})_{V\in
\mathcal{V(R)}}$ of $\underline{\Omega}$ is interpreted as the truth-value of
the proposition represented by $\underline{S}\in P_{cl}(\underline{\Sigma})$
if the quantum system is in the state $\psi$ (resp. $\mathbb{T}^{\psi}$). This
assignment of truth-values is

\begin{itemize}
\item contextual, since the contexts $V\in\mathcal{V(R)}$ play a central
r\^{o}le in the whole construction

\item global in the sense that \emph{every} proposition is assigned a truth-value

\item completely independent of any notion of measurement or observer, hence
we call our scheme a `neo-realist' formulation of quantum theory

\item topos-internal, the logical structure is not chosen arbitrarily, but
fixed by the topos $\operatorname*{\mathbf{Set}}^{\mathcal{V(R)}^{op}}$. This topos is
directly motivated from the Kochen-Specker theorem

\item non-Boolean, since there are (a) more truth-values than just
\emph{`true'} and \emph{`false'} and (b) the global elements form a
\emph{Heyting} algebra, not a Boolean algebra. There is a global element $1$
of $\underline{\Omega}$, consisting of the maximal sieve $\downarrow\!\!V$ at
each stage $V$, which is interpreted as `totally true', and there is a global
element $0$ consisting of the empty sieve for all $V$, which is interpreted as
`totally false'. Apart from that, there are many other global elements that
represent truth-values between `totally true' and `totally false'. These
truth-values are neither numbers nor probabilities, but are given by the
logical structure of the presheaf topos $\operatorname*{\mathbf{Set}}^{\mathcal{V(R)}%
^{op}}$. Since a Heyting algebra in particular is a partially ordered set,
there are truth-values $v_{1},v_{2}$ such that neither $v_{1}<v_{2}$ nor
$v_{2}<v_{1}$, which is also different from two-valued Boolean logic where
simply $0<1$ (i.e., `false'$<$`true'). The presheaf topos $\operatorname*{\mathbf{Set}}%
^{\mathcal{V(R)}^{op}}$ has a rich logical structure.
\end{itemize}

\section{Conclusion and outlook\label{_Conclusion}}

\ The formulation of quantum theory within the presheaf topos
$\operatorname*{\mathbf{Set}}^{\mathcal{V(R)}^{op}}$ gives a theory that is remarkably
similar to classical physics from a structural perspective. In particular,
there is a state object (the spectral presheaf $\underline{\Sigma}$) and a
quantity-value object (the presheaf $\underline{\mathbb{R}^{\leftrightarrow}}$
of order-preserving and -reversing functions). Physical quantities are
represented by arrows between $\underline{\Sigma}$ and $\underline
{\mathbb{R}^{\leftrightarrow}}$.

\ 

One of the future tasks will be the incorporation of dynamics. The process of
daseinisation behaves well with respect to the action of unitary operators,
see section 5.2 in \cite{DI07c}, so it is conceivable that there is a
`Heisenberg picture' of dynamics. Commutators remain to be understood in the
topos picture. On the other hand, it is possible to let a truth-object
$\mathbb{T}^{\psi}$ change in time by applying Schr\"{o}dinger evolution to
$\psi$. It remains to be shown how this can be understood topos-internally.

\ 

Mulvey and Banaschewski have recently shown how to define the Gel'fand
spectrum of an abelian $C^{\ast}$-algebra $\mathfrak{A}$ in any Grothendieck 
topos, using constructive methods (see \cite{BM06} and references therein). 
Spitters and Heunen made the following construction in \cite{HS07}: one takes 
a \emph{non-abelian} $C^{\ast}$-algebra $\mathfrak{A}$ and considers the topos 
of (covariant) functors over the category of abelian subalgebras of $\mathfrak{A}$. 
The algebra $\mathfrak{A}$ induces an internal \emph{abelian} 
$C^{\ast}$-algebra $\underline{\mathfrak{A}}$ in this topos of functors. 
(Internally, algebraic operations are only allowed between commuting operators.) 
Spitters and Heunen observed that the Gel'fand spectrum of this internal algebra 
basically is the spectral presheaf.\footnote{The change from presheaves, i.e., 
contravariant functors, to covariant functors is necessary in the constructive 
context.} It is very reassuring that the spectral presheaf not only has a physical 
interpretation, but also such a nice and natural mathematical one. Spitters and 
Heunen also discuss integration theory in the constructive context. These tools 
will be very useful in order to regain actual numbers and expectation values from 
the topos formalism.

\ 

Since the whole topos programme is based on the representation of formal
languages, major generalisations are possible. One can represent the same
language $\mathcal{L}(S)$ in different topoi, as we already did with
$\operatorname*{\mathbf{Set}}$ for classical physics and $\operatorname*{\mathbf{Set}}%
^{\mathcal{V(R)}^{op}}$ for algebraic quantum theory. For physical theories
going beyond this, other topoi will play a r\^{o}le. The biggest task is the
incorporation of space-time concepts, which will, at the very least,
necessitate a change of the base category. It is also conceivable that the
`smooth topoi' of synthetic differential geometry (SDG) will play a r\^{o}le.

\ 

\textbf{Acknowledgements.} I want to thank Chris Isham for numerous
discussions and constant support. Moreover, I would like to thank Bob Coecke,
Samson Abramsky, Jamie Vicary, Bas Spitters and Chris Heunen for stimulating
discussions. This work was supported by the DAAD and the FQXi, which is
gratefully acknowledged. Finally, I want to thank the organisers of the RDQFT
workshop, Bertfried Fauser, J\"{u}rgen Tolksdorf and Eberhard Zeidler, for
inviting me. It was a very enjoyable experience.


\begin{thebibliography}{999999}                                                                                           %


\bibitem[Bell88]{Bell88}J. L. Bell, \textit{Toposes and Local Set Theories}
(Clarendon Press, Oxford 1988)

\bibitem[BM06]{BM06}B. Banaschewski, C. J. Mulvey, \textquotedblleft A 
globalisation of the Gelfand duality theorem\textquotedblright, Ann. of Pure and 
Applied Logic \textbf{137} (2006), 62-103

\bibitem[BV36]{BV36}G. Birkhoff, J. von Neumann, \textquotedblleft The logic
of quantum mechanics\textquotedblright, Ann. of Math. \textbf{37} (1936), 823-843

\bibitem[Doe05]{Doe05}A. D\"{o}ring, \textquotedblleft Kochen-Specker Theorem
for von Neumann Algebras\textquotedblright, Int. J. Theor. Phys. \textbf{44},
no. 2 (2005), 139-160

\bibitem[DI07a]{DI07a}A. D\"{o}ring, C. J. Isham, \textquotedblleft A Topos
Foundation for Theories of Physics: I. Formal Languages for
Physics\textquotedblright, quant-ph/0703060

\bibitem[DI07b]{DI07b}A. D\"{o}ring, C. J. Isham, \textquotedblleft A Topos
Foundation for Theories of Physics: II. Daseinisation and the Liberation of
Quantum Theory\textquotedblright, quant-ph/0703062

\bibitem[DI07c]{DI07c}A. D\"{o}ring, C. J. Isham, \textquotedblleft A Topos
Foundation for Theories of Physics: III. The Representation of Physical
Quantities with Arrows $\delta^{o}(A):\underline{\Sigma}\rightarrow
\underline{\mathbb{R}^{\succeq}}$\textquotedblright, quant-ph/0703064

\bibitem[DI07d]{DI07d}A. D\"{o}ring, C. J. Isham, \textquotedblleft A Topos
Foundation for Theories of Physics: IV. Categories of
Systems\textquotedblright, quant-ph/0703066

\bibitem[Gol84]{Gol84}R. Goldblatt, \textit{Topoi, The Categorical Analysis of
Logic}, revised second edition (North-Holland, Amsterdam, New York, Oxford 1984)

\bibitem[deG05]{deG05}H. F. de Groote, \textquotedblleft On a canonical
lattice structure on the effect algebra of a von Neumann
algebra\textquotedblright, quant-ph/0410018 v2 (2005)

\bibitem[deG05b]{deG05b}H. F. de Groote, \textquotedblleft
Observables\textquotedblright, math-ph/0507019

\bibitem[deG07]{deG07}H. F. de Groote, \textquotedblleft Observables IV: The
Presheaf Perspective\textquotedblright, math-ph/0708.0677

\bibitem[HS07]{HS07}C. Heunen, B. Spitters, \textquotedblleft A topos for
algebraic quantum theory\textquotedblright, quant-ph/0709.4364

\bibitem[Ish97]{Ish97}C. J. Isham \textquotedblleft Topos Theory and
Consistent Histories: The Internal Logic of the Set of all Consistent
Sets\textquotedblright, Int. J. Theor. Phys. \textbf{36} (1997), 785-814

\bibitem[IB98]{IB98}C. J. Isham, J. Butterfield, \textquotedblleft A topos
perspective on the Kochen-Specker theorem: I. Quantum states as generalised
valuations\textquotedblright, Int. J. Theor. Phys. \textbf{37} (1998), 2669-2733

\bibitem[IB99]{IB99}C. J. Isham, J. Butterfield, \textquotedblleft A topos
perspective on the Kochen-Specker theorem: II. Conceptual aspects, and
classical analogues\textquotedblright, Int. J. Theor. Phys. \textbf{38}
(1999), 827-859

\bibitem[IB00]{IB00}C. J. Isham, J. Butterfield, \textquotedblleft Some
Possible Roles for Topos Theory in Quantum Theory and Quantum
Gravity\textquotedblright, Found. Phys. \textbf{30} (2000), 1707-1735

\bibitem[IBH00]{IBH00}C. J. Isham, J. Butterfield, J. Hamilton,
\textquotedblleft A topos perspective on the Kochen-Specker theorem: III. Von
Neumann algebras as the base category\textquotedblright, Int. J. Theor. Phys.
\textbf{39} (2000), 1413-1436

\bibitem[IB02]{IB02}C. J. Isham, J. Butterfield, \textquotedblleft A topos
perspective on the Kochen-Specker theorem: IV. Interval
valuations\textquotedblright, Int. J. Theor. Phys. \textbf{41}, no. 4 (2002), 613-639

\bibitem[Joh02a]{Joh02a}P. T. Johnstone, \textit{Sketches of an Elephant, A
Topos Theory Compendium}, Vol. 1 (Clarendon Press, Oxford 2002)

\bibitem[Joh02b]{Joh02b}P. T. Johnstone, \textit{Sketches of an Elephant, A
Topos Theory Compendium}, Vol. 2 (Clarendon Press, Oxford 2002)

\bibitem[KR83]{KR83}R. V. Kadison, J. R. Ringrose, \textit{Fundamentals of the
Theory of Operator Algebras}, Vol. 1 (Academic Press, San Diego 1983; AMS
reprint 2000)

\bibitem[KS67]{KS67}S. Kochen, E. P. Specker, \textquotedblleft The problem of
hidden variables in quantum mechanics\textquotedblright, Journal of
Mathematics and Mechanics \textbf{17} (1967), 59-87

\bibitem[LS86]{LS86}J. Lambek, P. J. Scott, \textit{Introduction to higher
order categorical logic} (Cambridge University Press, Cambridge, New York etc. 1986)

\bibitem[LR03]{LR03}F. W. Lawvere, R. Rosebrugh, \textit{Sets for
Mathematicians} (Cambridge University Press, Cambridge 2003)

\bibitem[McL71]{McL71}S. Mac Lane, \textit{Categories for the Working
Mathematician} (Springer, \ New York, Berlin etc. 1971; second, revised
edition 1997)

\bibitem[MM92]{MM92}S. Mac Lane, I. Moerdijk, \textit{Sheaves in Geometry and
Logic, A First Introduction to Topos Theory}, Springer (New York, Berlin etc. 1992)

\bibitem[Ols71]{Ols71}M. P. Olson, \textquotedblleft The Selfadjoint Operators
of a von Neumann Algebra Form a Conditionally Complete
Lattice\textquotedblright, Proc. of the AMS \textbf{28} (1971), 537-544

\bibitem[Sto36]{Sto36}M. H. Stone, \textquotedblleft The theory of
representations for Boolean algebras\textquotedblright, Trans. Amer. Math.
Soc. \textbf{40} (1936), 37-111
\end{thebibliography}
\end{document}